\documentclass[12pt, a4paper]{scrartcl}
\usepackage[svgnames,dvipsnames,rgb]{xcolor} 
\usepackage{amsthm}
\usepackage{amsmath}
\usepackage{amssymb}
\usepackage{mathtools}
\usepackage{algorithm}
\usepackage{algpseudocode}
\usepackage{caption}
\usepackage{subcaption}
\usepackage{tikz}
\usepackage{graphicx}
\usepackage{hyperref}
\usepackage[nameinlink]{cleveref}
\usepackage{orcidlink}
\hypersetup{
   breaklinks=true,colorlinks=true,linkcolor=colB!80!black,citecolor=DarkSeaGreen!80!black,urlcolor=colB!80!black
}

\theoremstyle{plain}
\newtheorem{lemma}{Lemma}[section]
\newtheorem{theorem}[lemma]{Theorem}
\newtheorem{remark}[lemma]{Remark}
\theoremstyle{definition}
\newtheorem{definition}[lemma]{Definition}

\title{Flexible realizations existence:\\ NP-completeness on sparse graphs and algorithms}

\author{Petr Laštovička\footnote{Faculty of Information Technology, Czech Technical University in Prague, Czechia,
  \texttt{\{petr.lastovicka, jan.legersky\}@fit.cvut.cz}} \and
Jan Legerský$^\ast$}

\date{}

\usetikzlibrary{calc}

\colorlet{ecol}{black!50!white}
\definecolor{colR}{rgb}{.932,.172,.172} 
\definecolor{colB}{rgb}{.255,.41,.884} 
\definecolor{colOrange}{RGB}{255,191,0} 

\definecolor{currentstroke1}{HTML}{BBBBBB}
\definecolor{currentstroke2}{HTML}{EE6677}
\definecolor{currentstroke3}{HTML}{228833}
\definecolor{currentstroke4}{HTML}{CCBB44}
\definecolor{currentstroke5}{HTML}{66CCEE}
\definecolor{currentstroke6}{HTML}{AA3377}
\definecolor{currentstroke7}{HTML}{4477AA}

\tikzstyle{vertex}=[circle, draw, fill=black, inner sep=0pt, minimum size=4pt]
\tikzstyle{smallvertex}=[vertex, minimum size=2pt]
\tikzstyle{vertexw}=[circle, draw=white, fill=white, inner sep=0pt, minimum size=2pt]
\tikzstyle{vertexSig}=[circle, draw, fill=colOrange, inner sep=0pt, minimum size=4pt]
\tikzstyle{edge}=[line width=1.5pt,ecol]
\tikzstyle{dots}=[dotted dash=1.5pt]
\tikzstyle{redge}=[edge,colR]
\tikzstyle{bedge}=[edge,colB]
\tikzstyle{yedge}=[edge,colOrange]
\tikzstyle{arrow}=[edge,line width=1pt,->]
\tikzstyle{diedge}=[edge,->]

\newcommand{\red}{\text{red}}
\newcommand{\blue}{\text{blue}}

\newcommand{\flexrilog}{\textsc{FlexRiLoG}}

\newcommand{\trcon}{$\triangle$-connected}
\newcommand{\trext}{$\triangle$-extended}
\newcommand{\trextshort}{$\triangle$-ext.}

\newcommand{\true}{\textit{True}}
\newcommand{\false}{\textit{False}}

\newcommand{\NN}{\mathbb{N}}

\DeclareMathOperator{\NAC}{NAC}
\newcommand{\nac}[1]{\NAC(#1)}

\DeclareMathOperator{\CP}{\Pi_{\NAC}}

\newcommand{\MSO}{\textbf{MSO}$_2$}
\DeclareMathOperator{\partition}{\mathbf{partition}}
\DeclareMathOperator{\cycle}{\mathbf{cycle}}
\DeclareMathOperator{\NACcond}{\mathbf{NACcond}}

\newcommand{\NaiveCycles}{\textsc{NaiveCycles}}
\newcommand{\None}{\textsc{None}}
\newcommand{\Neighbors}{\textsc{Neighbors}}
\newcommand{\NeighborsDegree}{\textsc{NeighborsDegree}}
\newcommand{\IsNACColoring}{\textsc{IsNACColoring}}
\newcommand{\MergeLinear}{\textsc{Linear}}
\newcommand{\SharedVertices}{\textsc{SharedVertices}}

\newcommand{\NACex}{$\exists$NAC}

\newcommand{\IntroduceVertexNode}[1]{\textsc{Introduce}$(#1)$}
\newcommand{\ForgetVertexNode}[1]{\textsc{Forget}$(#1)$}
\newcommand{\JoinNode}{\textsc{Join}}

\newcommand{\state}[8]{
  \begin{scope}[xshift=#1,yshift=#2]
    \node at (-0.6, 1) {$\biggl($};
    \node at (3.6, 1) {$\biggr)$};
    \node at (1.5, 0.1) {$,$};
    \node at (1.5, 2.8) {\footnotesize$#8$};
    \node[vertexw] (bottom) at (1.5, -0.5) {};
    \node[vertexw] (top) at (1.5, 3.3) {};
    \begin{scope}[opacity=0.2]
      \begin{scope}
        \node[vertex, fill=white, draw=white] (1r) at (0,2) {};
        \node[vertex, fill=white, draw=white] (2r) at (1,2) {};
        \node[vertex, fill=white, draw=white] (3r) at (1,1) {};
        \node[vertex, fill=white, draw=white] (4r) at (1,0) {};
        \node[vertex, fill=white, draw=white] (5r) at (0,0) {};
        \node[vertex, fill=white, draw=white] (6r) at (0,1) {};
      \end{scope}

      \begin{scope}[xshift=2cm]
        \node[vertex, fill=white, draw=white] (1b) at (0,2) {};
        \node[vertex, fill=white, draw=white] (2b) at (1,2) {};
        \node[vertex, fill=white, draw=white] (3b) at (1,1) {};
        \node[vertex, fill=white, draw=white] (4b) at (1,0) {};
        \node[vertex, fill=white, draw=white] (5b) at (0,0) {};
        \node[vertex, fill=white, draw=white] (6b) at (0,1) {};
      \end{scope}

			\foreach \v in {1r,2r,3r,4r,5r,6r, 1b,2b,3b,4b,5b,6b}{
				\node[smallvertex] at (\v) {};
			}
    \end{scope}
    \foreach \v in {#3}{
      \node[vertex] at (\v) {};
    }
    \foreach \u/\v in {#4}{
      \draw[redge] (\u) -- (\v);
    }
    \foreach \u/\v in {#5}{
      \draw[redge,dash pattern=on 2.3pt off 2.3pt] (\u) -- (\v);
    }
    \foreach \u/\v in {#6}{
      \draw[bedge] (\u) -- (\v);
    }
    \foreach \u/\v in {#7}{
      \draw[bedge,dash pattern=on 2.3pt off 2.3pt] (\u) -- (\v);
    }
  \end{scope}
}

\newcommand{\pathNAC}[3]{
  \begin{scope}[xshift=#1]
		\node[smallvertex, opacity=0.2] (1) at (0,2) {};
		\node[smallvertex, opacity=0.2] (2) at (1,2) {};
		\node[vertex] (3) at (1,1) {};
		\node[vertex] (4) at (1,0) {};
		\node[vertex] (5) at (0,0) {};
		\node[vertex] (6) at (0,1) {};

    \foreach \u/\v in {#2}{
      \draw[redge] (\u) -- (\v);
    }
    \foreach \u/\v in {#3}{
      \draw[bedge] (\u) -- (\v);
    }
  \end{scope}
}

\begin{document}

\maketitle

\begin{abstract}
	One of the questions in Rigidity Theory is whether a realization
	of the vertices of a graph in the plane is flexible, namely,
	if it allows a continuous deformation preserving the edge lengths.
	A flexible realization of a connected graph in the plane exists if and only if
	the graph has a NAC-coloring,
	which is a surjective edge coloring by two colors such that
	for each cycle, either all the edges have the same color, or there are at least two
	edges of each color.
	The question whether a graph has a NAC-coloring,
	and hence also the existence of a flexible realization,
	has been proven to be NP-complete.
	We show that this question is also NP-complete on graphs with maximum degree five
	and on graphs with the average degree
	at most $4+\varepsilon$ for every fixed $\varepsilon >0$.
	We also show that NAC-colorings can be counted in linear time for graphs with bounded treewidth.
	Since the only existing implementation of checking the existence of a NAC-coloring
	is rather naive, we propose new algorithms along with their implementation, which is significantly faster.
	We also focus on searching all NAC-colorings of a graph,
	since they provide useful information about its possible flexible realizations.
\end{abstract}

\section{Introduction}

A central object in Rigidity Theory is a (bar-joint) \emph{framework},
which is a graph $G$ with its $d$-dimensional \emph{realization} $p : V(G) \to \mathbb{R}^d$.
The framework $(G,p)$ is called \emph{flexible}
if it has a nontrivial \emph{flex}, that is a continuous path of $d$-dimensional realizations
$p_t$, $0 \le t \le 1$, with $p_0 = p$ such that for all $0 < t \leq 1$
we have $\|p_t(u)-p_t(v)\| = \|p(u)-p(v)\|$ for every edge $uv\in E(G)$,
but $\|p_t(u)-p_t(v)\| \neq \|p(u)-p(v)\|$ for some $u,v\in V(G)$.
Otherwise, the framework is called \emph{rigid}.

While for $d\geq 2$ it is NP-hard to decide whether a given $d$-dimensional framework
is rigid~\cite{Saxe1979} (see also~\cite{AbelDemainEtAl,Schaefer2013}),
the problem becomes more tractable, at least for ${d=2}$, when we consider \emph{generic}
behavior: given a graph, either almost all $d$-dimensional realizations are rigid,
or almost all $d$-dimensional realizations are flexible~\cite{generically_rigid_graphs}.
Hence, one can speak about (generically) \emph{$d$-rigid graphs} and (generically)
\emph{$d$-flexible} graphs.
A graph $G$ is \emph{minimally $d$-rigid} if it is $d$-rigid and
$G - e$ is $d$-flexible for every $e \in E(G)$.
Pollaczek-Geiringer~\cite{laman_original_1927}
and later independently Laman~\cite{laman_1970}
gave a characterization of minimally $2$-rigid graphs,
which yields a polynomial algorithm
for testing $2$-rigidity (see for instance~\cite{polynomial-min-rigid}),
while a combinatorial classification of $d$-rigid graphs for $d\geq 3$ remains an open problem.

In this paper, we focus on determining whether a graph admits a flexible realization.
We restrict ourselves to realizations $p$ that are \emph{quasi-injective}, namely,
$p(u) \neq p(v)$ for every edge $uv$.
Clearly, the question is mainly interesting for graphs that are generically rigid.
Searching for such ``paradoxical'' flexibility has a long history:
for instance Dixon~\cite{Dixon} gave two constructions of flexible injective $2$-dimensional realizations
of the complete bipartite graph $K_{3,3}$ at the end of the 19th century.
More than a hundred years later, these two constructions were shown
to be the only possible ones~\cite{WalterHusty}.

The question of the existence of a flexible realization is simple for $d = 1$,
since on the line every realization of a connected graph is $1$-rigid.
For $d \ge 3$, all non-complete graphs have a flexible injective $2$-dimensional realization
by placing two non-adjacent vertices arbitrarily and all other vertices on a line,
see for instance~\cite[Section~6]{GLS2019}.
The $2$-dimensional case, which is non-trivial, has been settled in~\cite{GLS2019}
using the following notion, see also \Cref{fig:3prism}.
\begin{definition}
	Let $\delta : E(G) \to \{\red, \blue\}$ be an edge coloring of a graph $G$.
	A cycle in $G$ is an \emph{almost red} cycle, if all edges of the
	cycle are red except one that is blue, and analogously for an \emph{almost blue} cycle.
	The coloring $\delta$ is called a \emph{NAC-coloring} if it is surjective
	and has no almost red nor almost blue cycles.
	Let $\nac{G}$ denote the set of all NAC-colorings of~$G$.
\end{definition}
\begin{theorem}[\cite{GLS2019}]
	\label{thm:NACmain}
	A connected graph $G$
	has a flexible quasi-injective $2$-dimensional realization if and only if
	it has a NAC-coloring.
\end{theorem}

\begin{figure}[ht]
	\centering
	\begin{tikzpicture}[rotate=0,scale=2]
		\node[vertex] (a) at (0,0) {};
		\node[vertex] (b) at (1,0) {};
		\node[vertex] (c) at (0.5,0.5) {};
		\node[vertex] (d) at (0,1.5) {};
		\node[vertex] (e) at (1,1.5) {};
		\node[vertex] (f) at (0.5,1) {};
		\draw[bedge] (a)edge(b) (b)edge(c) (c)edge(a) (d)edge(e) (e)edge(f) (f)edge(d) ;
		\draw[redge] (a)edge(d) (b)edge(e) (c)edge(f);
	\end{tikzpicture}
	\qquad
	\qquad
	\begin{tikzpicture}[rotate=0,scale=2]
		\node[vertex] (a) at (0,0) {};
		\node[vertex] (b) at (1,0) {};
		\node[vertex] (c) at (0.5,0.5) {};
		\node[vertex] (d) at (0,1) {};
		\node[vertex] (e) at (1,1) {};
		\node[vertex] (f) at (0.5,1.5) {};
		\draw[edge] (a)edge(b) (b)edge(c) (c)edge(a) (d)edge(e) (e)edge(f) (f)edge(d) (a)edge(d) (b)edge(e) (c)edge(f);
	\end{tikzpicture}
	\qquad
	\begin{tikzpicture}[rotate=0,scale=2]
		\node[vertex] (a) at (0,0) {};
		\node[vertex] (b) at (1,0) {};
		\node[vertex] (c) at (0.5,0.5) {};
		\node[vertex] (d) at ($(a) + (70:1cm)$) {};
		\node[vertex] (e) at ($(b) + (70:1cm)$) {};
		\node[vertex] (f) at ($(c) + (70:1cm)$) {};
		\draw[edge] (a)edge(b) (b)edge(c) (c)edge(a) (d)edge(e) (e)edge(f) (f)edge(d) (a)edge(d) (b)edge(e) (c)edge(f);
	\end{tikzpicture}
	\caption{The $3$-prism is generically $2$-rigid, but has flexible $2$-dimensional realizations (middle and right).
		It has a unique NAC-coloring modulo swapping colors (left).}
	\label{fig:3prism}
\end{figure}

While every disconnected graph obviously has a $2$-dimensional flexible injective realization,
it has a NAC-coloring if and only if it has at least two components which are not isolated vertices
or it has a component that has a NAC-coloring.
Therefore, the question whether a graph has a quasi-injective flexible realization
can be reduced to the problem
\begin{center}
	\NACex: \textit{Does a graph have a NAC-coloring?}
\end{center}

It was proven that \NACex{} for general graphs is
NP-complete by a reduction from the 3-SAT~\cite{np_complete}.
A related NP-complete problem is the existence of a stable cut,
where a~\emph{stable cut} $S$ of a graph $G$ is a vertex cut, i.e., $G - S$ is disconnected,
that is also stable (independent), i.e., no two vertices in $S$ are adjacent.
The relation is the following:
if a graph has a stable cut, it has also
a NAC-coloring~\cite[Theorem 4.4]{GLS2019}; but not the other way around,
since for instance the $3$-prism graph has no stable cut.
An immediate consequence is that if a graph has a vertex which is not in any 3-cycle,
then the graph has a NAC-coloring.
The existence of a stable cut
is NP-complete even on the class of graphs with maximum degree five~\cite{stable_cuts}
and on $n$-vertex graphs with at most $(2+\varepsilon)n$ edges~\cite{Le2003}.
We prove that \NACex{} is NP-complete for both of these classes as well.

By \Cref{thm:NACmain}, every connected $2$-flexible graph has a NAC-coloring.
Recently, the authors of~\cite{nac_minimally_rigid} showed that
every $2$-flexible graph actually has a stable cut, and it can be found in polynomial time.
They also proved that a minimally $2$-rigid graph $G$ has a NAC-coloring if and only if
$G$ is not a 2-tree (which can be checked in polynomial time).
Hence, for the class of graphs $G$ with $|E(G)|\leq 2|V(G)|-3$,
\NACex{} can be decided in~polynomial time,
since such a graph is either minimally $2$-rigid, or $2$-flexible.
Therefore, together with our results, the only cases for which the complexity of \NACex{}
remains unknown is for the class of graphs $G$ such that $|E(G)|\leq 2|V(G)|+c$, where $c\geq -2$ is a fixed constant.
In particular, it is open whether \NACex{} is NP-complete
on graphs of maximum degree four.

We also study the complexity of \NACex{} when parametrized by \emph{treewidth}
(which we recall in \Cref{def:tree_decomposition}).
Informally, a problem is \emph{fixed-parameter tractable} if there is a parameter $k$ (like treewidth or maximum degree)
and an algorithm whose running time is at most $f(k)\cdot n^c$ for some function $f$
not depending on the input size $n$ and a constant~$c$.
We show that \NACex{} is tractable when parametrized by treewidth.
We also give an explicit linear time algorithm counting the NAC-colorings of graphs with bounded treewidth.

NAC-colorings can be used also to formulate some sufficient or necessary conditions
on the existence of flexible injective realizations~\cite{GLSinjective},
to study all flexible injective realizations of some small rigid graphs~\cite{GLSclassification},
or to study flexibility of frameworks consisting of triangles and parallelograms~\cite{GL2024}.
\Cref{thm:NACmain} and some other results generalize also to infinite graphs~\cite{DLinfinite}.
In some of these cases, the knowledge of all NAC-colorings of a graph is needed.
Hence, one might want to compute all NAC-colorings besides checking the existence.

The package \flexrilog~\cite{flexrilog} can list all NAC-colorings of a (small) graph, but it is rather slow.
It exploits the following concept, which we use throughout the paper:
let $\triangle$ be a relation on the edge set of a graph $G$
such that $e_1 \triangle e_2$ if and only if there exists a 3-cycle $C$ in $G$
such that $e_1, e_2 \in E(C)$.
The graph $G$ is \emph{\trcon{}} if $E(G)$ is the only equivalence class
of the reflexive-transitive closure of $\triangle$.
A subgraph is a~\emph{\trcon{} component} if it is a maximal \trcon{} subgraph.
Since every 3-cycle has to be monochromatic in all NAC-colorings,
every \trcon{} component has to be monochromatic as well,
which gives a naive algorithm to compute all NAC-colorings used in \flexrilog.
In the second half of this paper, we focus on approaches that significantly improve this algorithm.
Our implementation is faster than \flexrilog{} by two orders of magnitude.

The paper is organized as follows:
the complexity results on \NACex{} in sparse graphs are presented in \Cref{sec:NP}.
We discuss fixed-parameter tractability when parametrized by treewidth in \Cref{sec:fpt}.
In \Cref{sec:alg} we propose various approaches how to compute NAC-colorings
faster than using the naive algorithm, and we compare them
by presenting benchmarks in \Cref{sec:benchmarks}.

\section{NP-completeness on sparse graphs}%
\label{sec:NP}

In this section we prove that \NACex{}, i.e., the question whether a NAC-coloring exists,
is NP-complete on the class of graphs with maximum degree five and also
on the class of graphs with average degree smaller than
$4+\varepsilon$ for any $\varepsilon > 0$.
We proceed by reduction from the 3-SAT problem,
namely, the question whether a boolean formula
in conjunctive normal form in which each clause contains exactly three literals is satisfiable.
The main idea of our proof is inspired by
the proof of NAC-coloring NP-completeness~\cite{np_complete}.
The maximum degree in~the constructed graph therein is linear in the number of variables.
We propose a different gadget construction that allows us to limit the maximum
degree in the constructed graph.

\begin{theorem}
	\label{theorem:nac-deg-5}
	The question whether a graph has a NAC-coloring is NP-complete
	on the class of graphs with maximum degree five.
\end{theorem}
\begin{proof}
	Let $\phi$ be a formula with variables $x_{1}, \dots, x_{n}$
	and clauses $L_1, \dots, L_k$.
	Our goal is to construct a graph $G_\phi$
	of size $O(n+k)$ such that $\phi$ is satisfiable if and only if
	$G_\phi$ has a NAC-coloring.

	We exploit the fact that \trcon{} components
	are monochromatic in every NAC-coloring.
	In particular, every subgraph isomorphic to a ladder graph with diagonals is monochromatic.
	We call a ladder graph such that every 4-cycle has one diagonal a~\emph{braced ladder}.
	We build a 2-tree structure called a \emph{train},
	which is a ``horizontal'' braced ladder with other ``vertical'' braced ladders
	glued so that the maximum degree is five, see \Cref{fig:proof_trains}.
	A train can be extended arbitrarily to connect more braced ladders.

	\begin{figure}[h]
		\centering
		\begin{tikzpicture}[scale=2]
			\node[vertex] (11) at (0.5, 0.5) {};
			\node[vertex] (12) at (0.5, 1.0) {};
			\node[vertex] (21) at (1.0, 0.5) {};
			\node[vertex] (22) at (1.0, 1.0) {};
			\node[vertex] (23) at (1.0, 1.5) {};
			\node[vertex] (31) at (1.5, 0.5) {};
			\node[vertex] (32) at (1.5, 1.0) {};
			\node[vertex] (33) at (1.5, 1.5) {};
			\node[vertex] (41) at (2.0, 0.5) {};
			\node[vertex] (42) at (2.0, 1.0) {};
			\node[vertex] (51) at (2.5, 0.5) {};
			\node[vertex] (52) at (2.5, 1.0) {};
			\node[vertex] (53) at (2.5, 1.5) {};
			\node[vertex] (61) at (3.0, 0.5) {};
			\node[vertex] (62) at (3.0, 1.0) {};
			\node[vertex] (63) at (3.0, 1.5) {};
			\node[vertex] (71) at (3.5, 0.5) {};
			\node[vertex] (72) at (3.5, 1.0) {};

			\draw[edge] (11)edge(21) (21)edge(31) (31)edge(41) (41)edge(51) (51)edge(61) (61)edge(71);
			\draw[edge] (12)edge(22) (22)edge(32) (32)edge(42) (42)edge(52) (52)edge(62) (62)edge(72);
			\draw[edge] (11)edge(12) (21)edge(22) (31)edge(32) (41)edge(42) (51)edge(52) (61)edge(62) (71)edge(72);
			\draw[edge] (12)edge(21) (22)edge(31) (31)edge(42) (42)edge(51) (51)edge(62) (61)edge(72);
			\draw[edge] (22)edge(23) (32)edge(33) (23)edge(33) (23)edge(32);
			\draw[edge] (52)edge(53) (62)edge(63) (53)edge(63) (52)edge(63);

			\node[] (11d) at (0.25, 0.5) {};
			\node[] (12d) at (0.25, 1.0) {};
			\node[] (71d) at (3.75, 0.5) {};
			\node[] (72d) at (3.75, 1.0) {};
			\draw[edge] (11d)edge(11) (12d)edge(12) (71d)edge(71) (72d)edge(72);
			\node[] (23d) at (1.0, 1.75) {};
			\node[] (33d) at (1.5, 1.75) {};
			\node[] (53d) at (2.5, 1.75) {};
			\node[] (63d) at (3.0, 1.75) {};
			\draw[edge] (23d)edge(23) (33d)edge(33) (53d)edge(53) (63d)edge(63);
		\end{tikzpicture}
		\qquad
		\begin{tikzpicture}[scale=2]
			\node[vertex] (11) at (0.5, 0.5) {};
			\node[vertex] (12) at (0.5, 1.0) {};
			\node[vertex] (21) at (1.0, 0.5) {};
			\node[vertex] (22) at (1.0, 1.0) {};
			\node[vertex] (31) at (1.5, 0.5) {};
			\node[vertex] (32) at (1.5, 1.0) {};
			\node[vertex] (41) at (2.0, 0.5) {};
			\node[vertex] (42) at (2.0, 1.0) {};

			\draw[edge] (11)edge(21) (21)edge(31) (31)edge(41);
			\draw[edge] (12)edge(22) (22)edge(32) (32)edge(42);
			\draw[edge] (11)edge(12) (21)edge(22) (31)edge(32) (41)edge(42);
			\draw[edge] (12)edge(21) (22)edge(31) (31)edge(42);

			\node[] (11d) at (0.25, 0.5) {};
			\node[] (12d) at (0.25, 1.0) {};
			\node[] (41d) at (2.25, 0.5) {};
			\node[] (42d) at (2.25, 1.0) {};
			\draw[edge] (11d)edge(11) (12d)edge(12) (41d)edge(41) (42d)edge(42);

			\begin{scope}[yshift=-1cm,xshift=-0.25cm]
				\node[] (ar) at (1.5, 1.25) {$\downarrow$};
				\node[vertex] (11) at (0.5, 0.5) {};
				\node[vertex] (12) at (0.5, 1.0) {};
				\node[vertex] (21) at (1.0, 0.5) {};
				\node[vertex] (22) at (1.0, 1.0) {};
				\node[vertexSig] (31) at (1.5, 0.5) {};
				\node[vertexSig] (32) at (1.5, 1.0) {};
				\node[vertex] (41) at (2.0, 0.5) {};
				\node[vertex] (42) at (2.0, 1.0) {};
				\node[vertex] (51) at (2.5, 0.5) {};
				\node[vertex] (52) at (2.5, 1.0) {};

				\draw[edge] (11)edge(21) (21)edge(31) (41)edge(51);
				\draw[edge] (12)edge(22) (22)edge(32) (42)edge(52);
				\draw[edge] (11)edge(12) (21)edge(22) (41)edge(42) (51)edge(52);
				\draw[edge] (12)edge(21) (22)edge(31) (41)edge(52);
				\draw[yedge] (31)edge(42) (31)edge(32) (31)edge(41) (32)edge(42);

				\node[] (11d) at (0.25, 0.5) {};
				\node[] (12d) at (0.25, 1.0) {};
				\node[] (51d) at (2.75, 0.5) {};
				\node[] (52d) at (2.75, 1.0) {};
				\draw[edge] (11d)edge(11) (12d)edge(12) (51d)edge(51) (52d)edge(52);
			\end{scope}
		\end{tikzpicture}
		\caption{A train (left) is formed by gluing braced ladders so that
			the maximum degree is five. The right figure shows how it can be extended.}
		\label{fig:proof_trains}
	\end{figure}

	We label the edges of $G_\phi$ with literals $x_i, \bar{x}_i$,
	where the bar denotes negation, for $1 \le i \le n$ and with $t, f$ literals.
	The edges in one \trcon{} component have always the same label.
	We construct the graph so that the edges with the same label
	have the same color in every NAC-coloring:
	eventually, we choose blue edges to correspond to \true{} and red edges to \false{}.

	We take $2n+2$ trains, the edges of each labeled by one literal, to which we will
	link other gadgets using braced ladders.
	Note that an edge of a graph such that its endvertices have degrees at most three and four,
	can be glued to a train via a braced ladder so that the maximum degree is at most five.

	For each variable $x_i$, we create two gadgets:
	one with cycle $A_i$ in the center
	with the edges of the cycle linked using braced ladders
	to the trains $x_i, \bar{x}_i$ and $t$,
	and the other linked to the trains $x_i, \bar{x}_i, t$ and $f$
	according to \Cref{fig:proof_enforce_true_false}.

	\begin{figure}[h]
		\centering
		\begin{tikzpicture}[scale=2.5]
			\node[vertex]      (22) at (1.25, 1.00) {};
			\node[]           (d22) at (1.00, 1.00) {};
			\node[vertex]      (23) at (1.25, 1.50) {};
			\node[]           (d23) at (1.00, 1.50) {};
			\node[vertex]      (32) at (1.50, 1.00) {};
			\node[vertex]      (33) at (1.50, 1.50) {};
			\node[vertex]      (42) at (2.00, 1.00) {};
			\node[vertex]      (43) at (2.00, 1.50) {};
			\node[vertexSig]   (44) at (2.00, 1.75) {};
			\node[vertex]      (45) at (2.00, 2.25) {};
			\node[vertex]      (46) at (2.00, 2.50) {};
			\node[]           (d46) at (2.00, 2.75) {};
			\node[vertex]      (52) at (2.50, 1.00) {};
			\node[vertex]      (53) at (2.50, 1.50) {};
			\node[vertexSig]   (54) at (2.50, 1.75) {};
			\node[vertex]      (55) at (2.50, 2.25) {};
			\node[vertex]      (56) at (2.50, 2.50) {};
			\node[]           (d56) at (2.50, 2.75) {};
			\node[vertex]      (62) at (3.00, 1.00) {};
			\node[vertex]      (63) at (3.00, 1.50) {};
			\node[vertex]      (72) at (3.25, 1.00) {};
			\node[]           (d72) at (3.50, 1.00) {};
			\node[vertex]      (73) at (3.25, 1.50) {};
			\node[]           (d73) at (3.50, 1.50) {};
			\node[vertex] (special) at (2.25, 0.75) {};

			\node[] at (2.25, 1.5) {$A_i$};

			\draw[edge] (32)edge(42) (33)edge(43) (32)edge(33) (32)edge(43);
			\draw[edge] (22)edge(32) (23)edge(33) (22)edge(23) (22)edge(33);
			\draw[edge] (32)edge(special) (42)edge(special) (33)edge(44) (43)edge(44) (42)edge(43);
			\draw[edge] (22)edge(d22) (23)edge(d23);
			\node[] at (1.0, 1.25) {$x_i$};
			\node[] at (2.125, 1.25) {$x_i$};

			\draw[edge] (52)edge(62) (53)edge(63) (62)edge(63) (53)edge(62);
			\draw[edge] (62)edge(72) (63)edge(73) (72)edge(73) (63)edge(72);
			\draw[edge] (62)edge(special) (52)edge(special) (63)edge(54) (53)edge(54) (52)edge(53);
			\draw[edge] (72)edge(d72) (73)edge(d73);
			\node[] at (3.50,  1.25) {$\bar{x}_i$};
			\node[] at (2.375, 1.25) {$\bar{x}_i$};

			\draw[bedge] (44)edge(45) (54)edge(55) (44)edge(55);
			\draw[bedge] (45)edge(46) (55)edge(56) (45)edge(56);
			\draw[bedge] (44)edge(54) (45)edge(55) (46)edge(56);
			\draw[bedge] (46)edge(d46) (56)edge(d56);
			\node[] at (2.25, 2.75)  {$t$};
			\node[] at (2.25, 1.875) {$t$};

			\begin{scope}[xshift=3.5cm]
				\node[vertex]    (13) at (0.75, 1.50) {};
				\node[]         (d13) at (0.50, 1.50) {};
				\node[vertex]    (14) at (0.75, 2.00) {};
				\node[]         (d14) at (0.50, 2.00) {};
				\node[vertex]    (23) at (1.00, 1.50) {};
				\node[vertex]    (24) at (1.00, 2.00) {};
				\node[vertex]    (31) at (1.50, 0.75) {};
				\node[]         (d31) at (1.50, 0.50) {};
				\node[vertex]    (32) at (1.50, 1.00) {};
				\node[vertexSig] (33) at (1.50, 1.50) {};
				\node[vertexSig] (34) at (1.50, 2.00) {};
				\node[vertex]    (35) at (1.50, 2.50) {};
				\node[vertex]    (36) at (1.50, 2.75) {};
				\node[]         (d36) at (1.50, 3.00) {};
				\node[vertex]    (41) at (2.00, 0.75) {};
				\node[]         (d41) at (2.00, 0.50) {};
				\node[vertex]    (42) at (2.00, 1.00) {};
				\node[vertexSig] (43) at (2.00, 1.50) {};
				\node[vertexSig] (44) at (2.00, 2.00) {};
				\node[vertex]    (45) at (2.00, 2.50) {};
				\node[vertex]    (46) at (2.00, 2.75) {};
				\node[]         (d46) at (2.00, 3.00) {};
				\node[vertex]    (53) at (2.50, 1.50) {};
				\node[vertex]    (54) at (2.50, 2.00) {};
				\node[vertex]    (63) at (2.75, 1.50) {};
				\node[]         (d63) at (3.00, 1.50) {};
				\node[vertex]    (64) at (2.75, 2.00) {};
				\node[]         (d64) at (3.00, 2.00) {};

				\draw[edge] (33)edge(34);
				\draw[edge] (43)edge(44);
				\draw[bedge] (34)edge(44);
				\draw[redge] (33)edge(43);
				\node[] at (1.750, 1.750) {$B_i$};
				\node[] at (1.375, 1.700) {$x_i$};
				\node[] at (2.125, 1.800) {$\bar{x}_i$};
				\node[] at (1.700, 2.125) {$t$};
				\node[] at (1.800, 1.375) {$f$};

				\draw[edge] (23)edge(33) (24)edge(34) (23)edge(24) (23)edge(34);
				\draw[edge] (13)edge(23) (14)edge(24) (13)edge(14) (13)edge(24);
				\draw[edge] (13)edge(d13) (14)edge(d14);
				\node[] at (0.50, 1.75) {$x_i$};

				\draw[edge] (43)edge(53) (44)edge(54) (53)edge(54) (43)edge(54);
				\draw[edge] (53)edge(63) (54)edge(64) (63)edge(64) (53)edge(64);
				\draw[edge] (63)edge(d63) (64)edge(d64);
				\node[] at (3.00, 1.75) {$\bar{x}_i$};

				\draw[bedge] (34)edge(35) (44)edge(45) (35)edge(45) (35)edge(44);
				\draw[bedge] (35)edge(36) (45)edge(46) (36)edge(46) (36)edge(45);
				\draw[bedge] (36)edge(d36) (46)edge(d46);
				\node[] at (1.75, 3.00) {$t$};

				\draw[redge] (32)edge(33) (42)edge(43) (32)edge(42) (33)edge(42);
				\draw[redge] (31)edge(32) (41)edge(42) (31)edge(41) (32)edge(41);
				\draw[redge] (31)edge(d31) (41)edge(d41);
				\node[] at (1.75, 0.50) {$f$};

			\end{scope}
		\end{tikzpicture}
		\caption{The gadgets for every variable $x_i$.
			For all variables together they enforce that the trains $x_i$ and $\bar{x}_i$ have different colors.}%
		\label{fig:proof_enforce_true_false}
	\end{figure}

	For each clause $L_i$, we create a gadget indicated
	in \Cref{fig:proof_clause} with cycle $C_i$ in the center.
	Let $\hat{x}_{i,1}, \hat{x}_{i,2}, \hat{x}_{i,3}$
	be literals used in $L_i$,
	where $\hat{x}_{i,j}$ denotes either $x_{i,j}$ or $\bar{x}_{i,j}$
	depending on~$L_i$. We link the edge labeled $\hat{x}_{i,j}$
	to the appropriate literal trains.
	Since each of the 3-prism subgraphs has only one NAC-coloring up to swapping colors,
	all edges labeled with the same literal have the same color in every NAC-coloring.

	\begin{figure}[h]
		\centering
		\begin{tikzpicture}[scale=2.5]
			\node[vertexSig] (35) at (1.5, 2.5) {};
			\node[vertexSig] (53) at (2.5, 1.5) {};
			\node[vertex]    (57) at (2.5, 3.5) {};
			\node[vertex]    (75) at (3.5, 2.5) {};
			\node[] at (2.5, 2.5) {$C_i$};

			\node[vertex] (13) at (0.5, 1.5) {};
			\node[vertex] (14) at (0.5, 2.0) {};
			\node[vertex] (23) at (1.0, 1.5) {};
			\node[vertex] (24) at (1.0, 2.0) {};
			\draw[bedge] (13)edge(23) (14)edge(24) (13)edge(14) (23)edge(24) (13)edge(24);
			\draw[bedge] (53)edge(35) (53)edge(23) (35)edge(24) (53)edge(24);
			\node[] (13d) at (0.25, 1.5 ) {};
			\node[] (14d) at (0.25, 2.0 ) {};
			\draw[bedge] (13)edge(13d) (14)edge(14d);
			\node[] at (0.25, 1.75) {$t$};
			\node[] at (2.125, 2.125) {$t$};

			\node[vertex]    (06)   at (0.25, 3.0 ) {};
			\node[vertex]    (07)   at (0.25, 3.5 ) {};
			\node[vertexSig] (16)   at (0.5 , 3.0 ) {};
			\node[vertexSig] (17)   at (0.5 , 3.5 ) {};
			\node[vertexSig] (26)   at (1.25, 3.0 ) {};
			\node[vertexSig] (27)   at (1.25, 3.5 ) {};
			\node[vertex]    (36)   at (1.5 , 3.0 ) {};
			\node[vertex]    (37)   at (1.5 , 3.5 ) {};
			\node[vertex]    (46)   at (2.0 , 3.0 ) {};
			\node[vertexSig] (p1m1) at (1.0,  3.25) {}; 
			\node[vertexSig] (p1m2) at (0.75, 3.25) {};
			\node[vertex]    (p1t1) at (1.0,  2.75) {}; 
			\node[vertex]    (p1t2) at (0.75, 2.75) {};
			\draw[edge] (35)edge(46) (46)edge(57) (57)edge(37) (36)edge(37) (46)edge(37);
			\draw[edge] (37)edge(27) (36)edge(26) (37)edge(26) (35)edge(36) (46)edge(36) (27)edge(26);
			\draw[bedge] (26)edge(16) (27)edge(17); 
			\draw[edge] (26)edge(p1m1) (27)edge(p1m1); 
			\draw[edge] (16)edge(p1m2) (17)edge(p1m2); 
			\draw[bedge] (p1m1)edge(p1m2) (p1t1)edge(p1t2);
			\draw[bedge] (p1m1)edge(p1t1) (p1m2)edge(p1t2) (p1m1)edge(p1t2);
			\draw[edge] (16)edge(17) (06)edge(07); 
			\draw[edge] (16)edge(06) (17)edge(07); 
			\draw[edge] (16)edge(07); 
			\node[] (06d) at (0.0 , 3.0) {};
			\node[] (07d) at (0.0 , 3.5) {};
			\draw[edge] (07)edge(07d) (06)edge(06d);
			\node[] (p1t1d) at (1.0 , 2.5) {};
			\node[] (p1t2d) at (0.75, 2.5) {};
			\draw[bedge] (p1t1)edge(p1t1d) (p1t2)edge(p1t2d);
			\node[] at (1.875, 2.625) {$\hat{x}_1$};
			\node[] at (2.375, 3.125) {$\hat{x}_1$};
			\node[] at (0.0  , 3.25 ) {$\hat{x}_1$};
			\node[] at (0.875, 2.625) {$t$};

			\node[vertex]    (66)   at (3.0 , 3.0 ) {};
			\node[vertex]    (76)   at (3.5 , 3.0 ) {};
			\node[vertex]    (77)   at (3.5 , 3.5 ) {};
			\node[vertexSig] (86)   at (3.75, 3.0 ) {};
			\node[vertexSig] (87)   at (3.75, 3.5 ) {};
			\node[vertexSig] (96)   at (4.5 , 3.0 ) {};
			\node[vertexSig] (97)   at (4.5 , 3.5 ) {};
			\node[vertex]    (A6)   at (4.75, 3.0 ) {};
			\node[vertex]    (A7)   at (4.75, 3.5 ) {};
			\node[vertexSig] (p2m1) at (4.0,  3.25) {}; 
			\node[vertexSig] (p2m2) at (4.25, 3.25) {};
			\node[vertex]    (p2t1) at (4.0,  2.75) {}; 
			\node[vertex]    (p2t2) at (4.25, 2.75) {};
			\draw[edge] (75)edge(66) (66)edge(57) (57)edge(77) (76)edge(77) (66)edge(77);
			\draw[edge] (77)edge(87) (76)edge(86) (77)edge(86) (75)edge(76) (66)edge(76) (87)edge(86);
			\draw[bedge] (86)edge(96) (87)edge(97); 
			\draw[edge] (86)edge(p2m1) (87)edge(p2m1); 
			\draw[edge] (96)edge(p2m2) (97)edge(p2m2); 
			\draw[bedge] (p2m1)edge(p2m2) (p2t1)edge(p2t2);
			\draw[bedge] (p2m1)edge(p2t1) (p2m2)edge(p2t2) (p2m1)edge(p2t2);
			\draw[edge] (96)edge(97) (A6)edge(A7); 
			\draw[edge] (96)edge(A6) (97)edge(A7); 
			\draw[edge] (96)edge(A7); 
			\node[] (A6d) at (5.0 , 3.0) {};
			\node[] (A7d) at (5.0 , 3.5) {};
			\draw[edge] (A7)edge(A7d) (A6)edge(A6d);
			\node[] (p2t1d) at (4.0 , 2.5) {};
			\node[] (p2t2d) at (4.25, 2.5) {};
			\draw[bedge] (p2t1)edge(p2t1d) (p2t2)edge(p2t2d);
			\node[] at (2.625, 3.125) {$\hat{x}_2$};
			\node[] at (3.125, 2.625) {$\hat{x}_2$};
			\node[] at (5.00 , 3.25 ) {$\hat{x}_2$};
			\node[] at (4.125, 2.625) {$t$};

			\node[vertex]    (64)   at (3.0 , 2.0 ) {};
			\node[vertex]    (74)   at (3.5 , 2.0 ) {};
			\node[vertex]    (73)   at (3.5 , 1.5 ) {};
			\node[vertexSig] (84)   at (3.75, 2.0 ) {};
			\node[vertexSig] (83)   at (3.75, 1.5 ) {};
			\node[vertexSig] (94)   at (4.5 , 2.0 ) {};
			\node[vertexSig] (93)   at (4.5 , 1.5 ) {};
			\node[vertex]    (A4)   at (4.75, 2.0 ) {};
			\node[vertex]    (A3)   at (4.75, 1.5 ) {};
			\node[vertexSig] (p3m1) at (4.0,  1.75) {}; 
			\node[vertexSig] (p3m2) at (4.25, 1.75) {};
			\node[vertex]    (p3t1) at (4.0,  2.25) {}; 
			\node[vertex]    (p3t2) at (4.25, 2.25) {};
			\draw[edge] (75)edge(64) (64)edge(53) (53)edge(73) (74)edge(73) (64)edge(73);
			\draw[edge] (73)edge(83) (74)edge(84) (73)edge(84) (75)edge(74) (64)edge(74) (83)edge(84);
			\draw[bedge] (84)edge(94) (83)edge(93); 
			\draw[edge] (84)edge(p3m1) (83)edge(p3m1); 
			\draw[edge] (94)edge(p3m2) (93)edge(p3m2); 
			\draw[bedge] (p3m1)edge(p3m2) (p3t1)edge(p3t2);
			\draw[bedge] (p3m1)edge(p3t1) (p3m2)edge(p3t2) (p3m1)edge(p3t2);
			\draw[edge] (94)edge(93) (A4)edge(A3); 
			\draw[edge] (94)edge(A4) (93)edge(A3); 
			\draw[edge] (94)edge(A3); 
			\node[] (A4d) at (5.0 , 2.0) {};
			\node[] (A3d) at (5.0 , 1.5) {};
			\draw[edge] (A3)edge(A3d) (A4)edge(A4d);
			\node[] (p3t1d) at (4.0 , 2.5) {};
			\node[] (p3t2d) at (4.25, 2.5) {};
			\draw[bedge] (p3t1)edge(p3t1d) (p3t2)edge(p3t2d);
			\node[] at (3.125, 2.375) {$\hat{x}_3$};
			\node[] at (2.625, 1.875) {$\hat{x}_3$};
			\node[] at (5.00 , 1.75 ) {$\hat{x}_3$};
			\node[] at (4.125, 2.375) {$t$};

		\end{tikzpicture}
		\caption{The gadget for the clause
			$(\hat{x}_{i,1} \lor \hat{x}_{i,2} \lor \hat{x}_{i,3})$, index $i$ omitted in the labels.}%
		\label{fig:proof_clause}
	\end{figure}

	Note that for the variable gadgets we add
	a fixed number of vertices and edges bounded by $O(n)$ and
	for the clause gadgets the number is bounded by $O(k)$.
	Therefore, the whole graph size is bounded by $O(n+k)$
	and the graph can be constructed in polynomial time.
	Also, note that the maximum degree is five.

	We prove that the graph $G_\phi$ has a NAC-coloring if and only if
	$\phi$ is satisfiable.
	First, suppose we have a NAC-coloring $\delta$ of $G_\phi$.
	Let the train $t$ be blue.
	We derive some properties of the NAC-coloring from the graph.
	We prove that the trains $x_i$ and~$\bar{x}_i$
	are colored with different colors for every $i$
	and the train $f$ is red.

	Assume for contradiction that the train $f$ is blue.
	Then the trains $x_i$ and $\bar{x}_i$ have the same color for all $i$,
	otherwise $B_i$ would form an almost cycle.
	Since every edge is labeled by a literal and NAC-coloring is surjective,
	there is literal $x_j$ such that the trains~$x_j$ and~$\bar{x}_j$ are red.
	But then the cycle $A_j$ is an almost cycle, which is a contradiction.
	Hence, the train $f$ is red.
	From the cycles $B_i$ we also see
	that trains $x_i$ and $\bar{x}_i$ have to be colored with different colors
	for every $i$.

	Now we create the related truth assignment.
	For each variable $x_i$ we assign \true{} if
	the train $x_i$ is blue,
	otherwise \false{}.
	Each clause $L_i$ is satisfied since
	in every cycle $C_i$, at least one of
	the literals $\hat{x}_{i,j}$ corresponds to blue colored
	edges, otherwise an almost red cycle is formed.
	Therefore, a truth assignment for $\phi$ can be obtained
	from a NAC-coloring of~$G_\phi$ in polynomial time.

	Now we prove that for every truth assignment such that $\phi$ evaluates to \true{}, there exist
	a NAC-coloring of $G_\phi$. We define an edge coloring
	$\delta: E(G_\phi) \to \{\red, \blue\}$ as follows:
	the edges labeled with $t$ and $f$ are blue and red respectively,
	and an edge labeled by literals $x_i$, resp.\ $\bar{x}_i$, is blue
	if $x_i$, resp.\ $\bar{x}_i$, evaluates to \true{} in the truth assignment, and red otherwise.
	Since $t$ and $f$ have different colors,
	the coloring $\delta$ is surjective.

	Suppose there is an almost cycle $C$.
	Let $e=uv$ be the edge of the almost cycle $C$ that has the opposite color
	than all the other edges of $C$.
	The vertices $u$ and $v$ must be contained in multiple \trcon{} components
	since these are monochromatic.
	In the gadgets, all such possibilities for $e$ are indicated by edges with yellow endvertices.
	The edge~$e$ cannot be in any train since there are no two adjacent vertices that are also
	in~some other \trcon{} component.

	Now we use the fact that both $u$ and $v$ must be incident to edges of both colors.
	The gadget in \Cref{fig:proof_enforce_true_false} with cycle $A_i$
	cannot contain $e$ since exactly one of the two yellow vertices is incident only to blue edges
	as either $x_i$ or $\bar{x}_i$ are colored blue.
	The edge $e$ is not in~the gadget with cycle $B_i$ either,
	since there is a pair of opposite vertices of the cycle $B_i$
	such that one vertex is incident only to blue edges and the other one only to red ones
	as $x_i$ is blue and $\bar{x}_i$ is red, or the other way around.

	For the third gadget in \Cref{fig:proof_clause},
	suppose first that $e$ is the edge in cycle $C_i$ labeled $t$.
	Since it is blue, the edges labeled $\hat{x}_{i,1}$ and $\hat{x}_{i,3}$ are red.
	The edges labeled $\hat{x}_{i,2}$ are blue since the $i$-th  clause evaluates to \true{}.
	Hence, neither $C_i$ nor any other cycle inside the gadget is an almost red cycle.
	Since every cycle containing $e$ that does not lie entirely in the gadget
	has to pass through some of the 3-prism subgraphs in the gadget, it contains another blue edge labeled by $t$.
	An analogous argument applies also for $e$ being any other edge in the gadget labeled by $t$.
	It remains to consider the case when $e$ is one of the edges in~a~3-cycle of the 3-prism subgraphs.
	But in this case $C$ cannot be an almost cycle either since both triangles in each 3-prism are colored the same.
\end{proof}

Let us present a case in which 3-prisms are needed in the construction.
Let $\phi = (A \lor B \lor C) \land (A \lor D \lor B)$.
For~each satisfiable truth assignment, there is a~NAC-coloring in~$G_\phi$.
This formula is satisfiable, for~example, if~$C$ and~$D$ are assigned~$\true$
and all the~other literals are assigned~$\false$.
If~we create the corresponding $\red$-$\blue$-coloring in~$G_\phi$
where the~clause gadgets do not have the 3-prisms, an~almost cycle is created:
the~cycle starts at~$A$ in~the~first clause, goes through connecting
train to the $A$ section in~the~second clause. There, it~uses the~$\blue$
$\true$ edge and goes to~$B$. Using a~connecting train,
it~goes back to the~$B$ segment in~the~first clause.
Here, it~joins back to~$A$ as the~segments share a~vertex.
All edges corresponding to~$A$ and~$B$
are $\red$, and we~used only a~single $\blue$ edge.
Therefore, an~almost cycle exists, and the~coloring is not a~NAC-coloring.
If~the~prisms are used, four more blue edges are visited in an~analogous cycle.

\begin{theorem}
	\label{theorem:nac-eps}
	For every $\varepsilon>0$,
	the question whether a NAC-coloring exists is NP-complete for the class of graphs $G$
	with $|E(G)| \leq (2 + \varepsilon) |V(G)|$.
\end{theorem}
\begin{proof}
	Fix $\varepsilon>0$.
	By \Cref{theorem:nac-deg-5}, we can assume $\varepsilon<\frac{1}{2}$,
	since $|E(G)| \leq \frac{5}{2} |V(G)|$ if a graph $G$ has maximum degree five.
	The proof of the previous theorem applies once we show that
	we can create a graph~$G'_\phi$ from the graph $G_\phi$ constructed for a formula $\phi$
	such that $|E(G'_\phi)| \leq (2 + \varepsilon) |V(G'_\phi)|$, and
	the graph $G'_\phi$ has a NAC-coloring if and only if $G_\phi$ has a NAC-coloring.
	We take any braced ladder in a train of $G_\phi$
	and extend it $k$ times as indicated in \Cref{fig:proof_trains}.
	Hence, $|E(G'_\phi)| = |E(G_\phi)|+4k$ and $|V(G'_\phi)| = |V(G_\phi)|+2k$.
	Since we modify only a \trcon{} component,
	there is a bijection between NAC-colorings of $G'_\phi$ and $G_\phi$.
	Using the fact that the maximum degree of $G_\phi$ is five,
	we have for any integer $k > \frac{1-2\varepsilon}{4\varepsilon}|V(G_\phi)|$
	that
	\[
		\frac{|E(G_\phi')|}{|V(G_\phi')|}= \frac{|E(G_\phi)| + 4k}{|V(G_\phi)| + 2k}
		\leq \frac{\frac{5}{2}|V(G_\phi)| + 4k}{|V(G_\phi)| + 2k}
		= 2 + \frac{\frac{1}{2}|V(G_\phi)|}{|V(G_\phi)| + 2k}
		< 2+ \varepsilon\,. \qedhere
	\]
\end{proof}
\newpage

\section{Fixed-parameter tractability by treewidth}
\label{sec:fpt}

In this section we show in two different ways that \NACex{}
is fixed-parameter tractable when parameterized by treewidth
(see~\cite{paramAlg} for more on these notions).
The first one is a direct application of Courcelle's theorem~\cite{Courcelle},
while the second approach gives a finer asymptotic bound on the complexity of counting the number of NAC-colorings.

We start by recalling the definition of treewidth, see \Cref{fig:nice_tree_decomposition} for an example.
\begin{definition}%
	\label{def:tree_decomposition}
	A \emph{tree decomposition} of a graph~$G$ is
	a pair $(T, {\{X_t\}}_{t \in V ( T )})$,
	where $T$ is a tree and every node $t\in V(T)$
	is assigned a \emph{bag} $X_t \subseteq V(G)$ such that:
	\begin{enumerate}
		\item $\bigcup_{t \in V(T)} X_t = V(G)$,
		      i.e., each vertex of $G$ is in at least one bag,
		\item for every $uv \in E(G)$, there exists
		      a node $t \in T$ such that both $u, v \in X_t$, and
		\item for every $u \in V(G)$, the subgraph of $T$ induced by $\{t \in V(T) | u \in X_t\}$
		      is connected.
	\end{enumerate}
	The \emph{width} of the tree decomposition
	is $\max_{t\in V(T)} |X_t| - 1$.
	The \emph{treewidth} of~$G$ is the minimum such width
	across all tree decompositions of~$G$.
	We call the elements of~$V(T)$ \emph{nodes}
	and write $t \in T$ to abbreviate $t \in V(T)$.
\end{definition}

By Courcelle's theorem~\cite{Courcelle}, a problem is fixed-parameter tractable
when parametrized by treewidth if it can be expressed
in the monadic second-order logic (\MSO), see for instance~\cite[Section~7.4.1]{paramAlg}.
In \MSO{}, we can quantify over vertices, edges and subsets of vertices and edges.
Let $\cycle(C)$ be a formula encoding that the (sub)graph induced by edges $C$ is a cycle.
This is equivalent to being connected and the condition that every vertex has degree two,
which both can be expressed in \MSO, see for instance~\cite[Section~7.4.1]{paramAlg}.
Now, we can formalize the existence of a NAC-coloring of a graph $G=(V,E)$ as follows:
\begin{align*}
	(\exists E_r,E_b \subseteq E)\big(\partition(E_b,E_r) \land (\forall C\subseteq E)(\cycle(C) \implies \NACcond(C,E_b,E_r))\big),
\end{align*}
where $\partition(E_b,E_r)$ is defined by
\begin{align*}
	(\exists e_1,e_2 \in E)(e_1 \in E_b \land e_2\in E_r)\land(\forall e \in E)\big((e\in E_b \lor e\in E_r) \land (e\notin E_b \lor e\notin E_r)\big)\,,
\end{align*}
and $\NACcond(C,E_b,E_r)$ by
\begin{align*}
	C\subseteq E_b \lor C\subseteq E_r \lor (\exists e_1,e_2,e_3,e_4\in C)	(e_1\neq e_2 \land e_3 \neq e_4 \land e_1,e_2 \in E_b \land e_3,e_4\in E_r)\,.
\end{align*}
Hence, the existence of a NAC-coloring is fixed-parameter tractable by treewidth.

In order to get an explicit bound on the complexity in terms of treewidth, we describe an algorithm
taking on the input besides a graph also its special tree decomposition,
where we restrict how bags of neighboring nodes can differ and assume the tree is rooted.
\begin{definition}
	A tree decomposition $(T, {\{X_t\}}_{t \in V ( T )})$ rooted at~$r \in T$
	is \emph{nice} if $|X_r| = 1$, $|X_l| = 1$ for every leaf node $l \in T$
	and every non-leaf node is one of the following types:
	\begin{itemize}
		\item \IntroduceVertexNode{v} --- a node $t$ with one child $t'$
		      such that $X_t = X_{t'} \cup \{v\}$, where $v \not\in X_{t'}$.
		      We say that $v$ is \emph{introduced} by $t$.
		\item \ForgetVertexNode{v} --- a node $t$ with one child $t'$
		      such that $X_t = X_{t'} \setminus \{v\}$, where $v \in X_{t'}$.
		\item \JoinNode{} --- a node $t$ with two children $t_1, t_2$
		      such that $X_t = X_1 = X_2$.
	\end{itemize}
	For $t \in T$, we denote the set of vertices introduced by $t$
	and all its child nodes by $V_t$, and $G_t$ is the subgraph of $G$ induced by $V_t$.
\end{definition}
Note that contrary to the definition of nice tree decompositions in~\cite{paramAlg},
it is more convenient for us to consider the root and leaf vertices to contain a single vertex.
An~example of a nice tree decomposition is shown in~\Cref{fig:nice_tree_decomposition}.

\begin{figure}[ht]
	\begin{center}
		\begin{tikzpicture}
			\begin{scope}[xshift=-7cm, yshift=1cm]
				\begin{scope}[yshift=4cm, xshift=0.5cm, scale=2]
					\node[vertex,label={left:$1$}] (1) at (0,2) {};
					\node[vertex,label={right:$2$}] (2) at (1,2) {};
					\node[vertex,label={right:$3$}] (3) at (1,1) {};
					\node[vertex,label={right:$4$}] (4) at (1,0) {};
					\node[vertex,label={left:$5$}] (5) at (0,0) {};
					\node[vertex,label={left:$6$}] (6) at (0,1) {};
					\draw[edge] (1)--(2) (2)--(3) (1)--(3) (3)--(4) (4)--(5) (5)--(6) (6)--(1) ;
				\end{scope}

				\begin{scope}[scale=2]
					\node (A) at (0,0) {$\{1,2,3\}$};
					\node (B) at (0,1) {$\{1,3,6\}$};
					\node (C) at (1.5,1) {$\{3,4,6\}$};
					\node (D) at (1.5,0) {$\{4,5,6\}$};
					\draw[edge] (A)--(B) (B)--(C) (C)--(D);
				\end{scope}
			\end{scope}

			\begin{scope}[yscale=1.3]
				\begin{scope}
					\node[align=center] (intr1) at (0, 0) {$\{1\}$, \IntroduceVertexNode{1}};
					\node[align=center] (intr2) at (0, 1) {$\{1, 2\}$, \IntroduceVertexNode{2}};
					\node[align=center] (intr3) at (0, 2) {$\{1,2,3\}$, \IntroduceVertexNode{3}};
					\node[align=center] (forg2) at (0, 3) {$\{1,3\}$, \ForgetVertexNode{2}};
					\node[align=center] (intr6) at (0, 4) {$\{1,3,6\}$, \IntroduceVertexNode{6}};
					\node[align=center] (forg1) at (0, 5) {$\{3,6\}$, \ForgetVertexNode{1}};
				\end{scope}

				\begin{scope}[xshift=5cm]
					\node[align=center] (forg4) at (0, 5) {$\{3,6\}$, \ForgetVertexNode{4}};
					\node[align=center] (intr3b) at (0, 4) {$\{3,4,6\}$, \IntroduceVertexNode{3}};
					\node[align=center] (forg5) at (0, 3) {$\{4,6\}$, \ForgetVertexNode{5}};
					\node[align=center] (intr6b) at (0, 2) {$\{4,5,6\}$, \IntroduceVertexNode{6}};
					\node[align=center] (intr5) at (0, 1) {$\{4,5\}$, \IntroduceVertexNode{5}};
					\node[align=center] (intr4) at (0, 0) {$\{4\}$, \IntroduceVertexNode{4}};
				\end{scope}

				\begin{scope}[xshift=2.5cm,yshift=6cm]
					\node[align=center] (join) at (0,0) {$\{3,6\}$, \JoinNode};
					\node[align=center] (forg6) at (0,1) {$\{3\}$, \ForgetVertexNode{6}};
				\end{scope}
			\end{scope}

			\draw[diedge] (intr1)--(intr2);
			\draw[diedge] (intr2)--(intr3);
			\draw[diedge] (intr3)--(forg2);
			\draw[diedge] (forg2)--(intr6);
			\draw[diedge] (intr6)--(forg1);
			\draw[diedge] (forg1)--(join);
			\draw[diedge] (forg4)--(join);
			\draw[diedge] (intr3b)--(forg4);
			\draw[diedge] (forg5)--(intr3b);
			\draw[diedge] (intr6b)--(forg5);
			\draw[diedge] (intr5)--(intr6b);
			\draw[diedge] (intr4)--(intr5);
			\draw[diedge] (join)--(forg6);
		\end{tikzpicture}
	\end{center}
	\caption{
		A graph (top left) with a tree decomposition (bottom left)
		and a nice tree decomposition with indicated types of nodes (right), both of width two.
	}%
	\label{fig:nice_tree_decomposition}
\end{figure}

Any tree decomposition of width at~most $k$ can be converted to
a nice tree decomposition of width at~most $k$
in time $O(k^2 \cdot \max(|V(T)|, |V(G)|))$, see \cite[Lemma~7.4]{paramAlg}.
The nice decomposition tree has at~most $O(k|V(G)|)$ nodes.
The problem of finding treewidth is NP-complete~\cite{tree_width_np_complete},
but there exist efficient approximation algorithms~\cite{tree_width_approximation}.
We consider a nice tree decomposition as given, along with a~graph,
and do not consider runtime required to find it.

In the rest of this section, we prove this more refined statement.
\begin{theorem}
	\label{thm:treewidth}
	Given a nice tree decomposition $T$ of a graph $G$,
	the number of NAC-colorings of $G$ can be found in time
	$O(k^{5k} \cdot 4^{k^2} \cdot |V(G)|)$,
	where $k$ is the width of $T$.
\end{theorem}

For the following, we fix a nice tree decomposition $(T, {\{X_t\}}_{t \in V ( T )})$  of a graph $G$.
The idea is to compute the number of NAC-colorings of graphs $G_t$
sorted into groups according to how vertices in $X_t$ are connected by (almost) red or blue paths,
which is captured by \emph{states} defined below.
This allows us to compute the numbers in a bottom-up manner.

\begin{definition}
	\label{def:state}
	Let $t\in T$.
	A \emph{state} is formed by a pair $(H_\red, H_\blue)$
	of two graphs with edges of two possible types
	such that $V(H_\red) = V(H_\blue) = X_t$ and for both $H_\red$
	and~$H_\blue$, the following holds:
	\begin{itemize}
		\item each edge is either \emph{solid} or \emph{dashed},
		\item if vertices $u$ and $v$ are connected by a path consisting of solid edges,
		      then $uv$ is a solid edge, and
		\item if vertices $u$ and $v$ are connected by a path that contains exactly one dashed edge,
		      then $uv$ is a dashed edge.
	\end{itemize}
	All such states on $X_t$ form the \emph{state space} $S_t$.
\end{definition}

Besides NAC-colorings, we need to consider also monochromatic colorings
as they do not contain any almost monochromatic cycle.
\begin{definition}
	An edge coloring by red and blue is \emph{without almost cycles} if it is a~NAC-coloring,
	or all edges are blue, or all edges are red.
	For such a coloring, a path is called \emph{almost red} if exactly one edge
	is blue, similarly for \emph{almost blue}.
	We say that a path is \emph{almost monochromatic} if it is almost red or almost blue.
\end{definition}
Notice that a red edge forms an almost blue path.
Since no two vertices can be connected by a monochromatic and
almost monochromatic path of the same color simultaneously, the following holds.

\begin{remark}
	\label{rem:consistent_state}
	For $t \in T$ and a coloring without almost cycles $\gamma$ of $G_t$,
	there exists a~unique state $(H_\red, H_\blue) \in S_t$ such that
	\begin{itemize}
		\item for $u, v \in X_t$, there is solid edge $uv$ in $H_\red$, resp.\ $H_\blue$,
		      if and only if there is a $\red$, resp. $\blue$, path from $u$ to $v$ in $(G_t, \gamma)$, and
		\item for $u, v \in X_t$, there is dashed edge $uv$ in $H_\red$, resp.\ $H_\blue$,
		      if and only if there is an~almost $\red$, resp. $\blue$, path from $u$ to $v$ in $(G_t, \gamma)$.
	\end{itemize}
	This state is called \emph{consistent} with $\gamma$.
\end{remark}
Notice that if $X_t=\{v\}$, then all colorings without almost cycles of $G_t$ are consistent
with the only state in $S_t$, namely $\left((\{v\}, \emptyset), (\{v\}, \emptyset)\right)$.
All the introduced concepts including the following definition are illustrated in \Cref{fig:states_and_colorings}.

\begin{definition}
	Let $\mathcal{S} = \{ (t,s) \mid t \in T, s \in S_t \}$.
	We define \emph{cache function} $c: \mathcal{S} \to \NN_0$ to be the function that maps $(t, s)$
	to the number of colorings without almost cycles of~$G_t$ consistent with $s$.
\end{definition}

\begin{figure}[ht]
	\centering
	\begin{tikzpicture}[scale=0.4]
		\begin{scope}[xshift=-6cm,yshift=-3cm,scale=2.5]
			\node[vertex,label={left:$1$}] (1) at (0,2) {};
			\node[vertex,label={right:$2$}] (2) at (1,2) {};
			\node[vertex,label={right:$3$}] (3) at (1,1) {};
			\node[vertex,label={right:$4$}] (4) at (1,0) {};
			\node[vertex,label={left:$5$}] (5) at (0,0) {};
			\node[vertex,label={left:$6$}] (6) at (0,1) {};
			\draw[edge] (1)--(2) (2)--(3) (1)--(3) (3)--(4) (4)--(5) (5)--(6) (6)--(1) ;
		\end{scope}

		\state{0cm}{0cm}{3r, 3b, 4r, 4b, 6r, 6b}{4r/6r, 3r/4r, 3r/6r}{}{}{3b/4b}{1}
		\state{5.5cm}{0cm}{3r, 3b, 4r, 4b, 6r, 6b}{4r/6r}{3r/4r, 3r/6r}{3b/4b}{}{1}
		\state{11cm}{0cm}{3r, 3b, 4r, 4b, 6r, 6b}{3r/4r}{4r/6r, 3r/6r}{}{4b/6b, 3b/4b}{2}
		\state{16.5cm}{0cm}{3r, 3b, 4r, 4b, 6r, 6b}{}{4r/6r, 3r/4r}{3b/4b}{4b/6b, 3b/6b}{2}
		\state{22cm}{0cm}{3r, 3b, 4r, 4b, 6r, 6b}{3r/4r}{}{4b/6b}{3b/4b, 3b/6b}{1}
		\state{27.5cm}{0cm}{3r, 3b, 4r, 4b, 6r, 6b}{}{3r/4r}{4b/6b, 3b/4b, 3b/6b}{}{1}

		\begin{scope}[yshift=-3.5cm]
			\pathNAC{1cm}{3/4,4/5,5/6}{}
			\pathNAC{6.5cm}{4/5,5/6}{3/4}
			\pathNAC{11cm}{3/4,4/5}{5/6}
			\pathNAC{13cm}{3/4,5/6}{4/5}
			\pathNAC{16.5cm}{5/6}{3/4,4/5}
			\pathNAC{18.5cm}{4/5}{3/4,5/6}
			\pathNAC{23cm}{3/4}{4/5,5/6}
			\pathNAC{28.5cm}{}{3/4,4/5,5/6}
		\end{scope}
	\end{tikzpicture}
	\caption{Let $t$ be the \IntroduceVertexNode{3} node with bag $\{3,4,6\}$
		in the nice tree decomposition in \Cref{fig:nice_tree_decomposition}.
		The graph $G_t$ is the path $(3,4,5,6)$.
		The colorings without almost cycles of $G_t$ are displayed together
		with their consistent states in $S_t$.
		The value of the cache function is above the states,
		those whose cache function is zero are not shown.
		Instead of labeling the vertices in the bag in each state, we display them in black
		and keep the same position as in the graph on the left, similarly for the colorings without almost cycles.
		The small grey dots are only auxiliary to make the positions easily identifiable.
	}
	\label{fig:states_and_colorings}
\end{figure}

Since the number of all colorings without almost cycles in $G_t$ is
\[
	\sum_{s \in S_t} c[t, s]\,,
\]
and $G=G_r$, where $r$ is the root of $T$, the goal is to show how $c$
can be computed recursively, see \Cref{fig:fpt_example} for an illustration on an example.
The easiest situation is when there are no edges.

\begin{figure}
	\centering
	\begin{tikzpicture}[scale=0.4]
		\begin{scope}[xshift=1cm,yshift=-4cm,scale=2.5]
			\node[vertex,label={left:$1$}] (1) at (0,2) {};
			\node[vertex,label={right:$2$}] (2) at (1,2) {};
			\node[vertex,label={right:$3$}] (3) at (1,1) {};
			\node[vertex,label={right:$4$}] (4) at (1,0) {};
			\node[vertex,label={left:$5$}] (5) at (0,0) {};
			\node[vertex,label={left:$6$}] (6) at (0,1) {};
			\draw[edge] (1)--(2) (2)--(3) (1)--(3) (3)--(4) (4)--(5) (5)--(6) (6)--(1) ;
		\end{scope}

		\node (intr1) at (33, 1) {\IntroduceVertexNode{1}};

		\node (intr2) at (22, 1) {\IntroduceVertexNode{2}};
		\draw[diedge] (intr1)--(intr2);
		\state{8cm}{0cm}{2r, 2b, 1r, 1b}{2r/1r}{}{}{2b/1b}{1}
		\node[vertexw] (bottom1) at (bottom) {};
		\state{13cm}{0cm}{2r, 2b, 1r, 1b}{}{2r/1r}{2b/1b}{}{1}
		\node[vertexw] (bottom2) at (bottom) {};

		\node (intr3) at (22, -3.5) {\IntroduceVertexNode{3}};
		\draw[diedge] (intr2)--(intr3);
		\state{8cm}{-4.5cm}{1r, 1b, 2r, 2b, 3r, 3b}{1r/2r, 2r/3r, 1r/3r}{}{}{1b/2b, 2b/3b, 1b/3b}{1}
		\draw[arrow] (bottom1)--(top);
		\node[vertexw] (bottom1) at (bottom) {};
		\state{13cm}{-4.5cm}{1r, 1b, 2r, 2b, 3r, 3b}{}{1r/2r, 2r/3r, 1r/3r}{1b/2b, 2b/3b, 1b/3b}{}{1}
		\draw[arrow] (bottom2)--(top);
		\node[vertexw] (bottom2) at (bottom) {};

		\node (forg2) at (22, -8) {\ForgetVertexNode{2}};
		\draw[diedge] (intr3)--(forg2);
		\state{8cm}{-9cm}{1r, 1b, 3r, 3b}{1r/3r}{}{}{1b/3b}{1}
		\draw[arrow] (bottom1)--(top);
		\node[vertexw] (bottom1) at (bottom) {};
		\state{13cm}{-9cm}{1r, 1b, 3r, 3b}{}{1r/3r}{1b/3b}{}{1}
		\draw[arrow] (bottom2)--(top);
		\node[vertexw] (bottom2) at (bottom) {};

		\node (intr6) at (24, -12.5) {\IntroduceVertexNode{6}};
		\draw[diedge] (forg2)--(intr6);
		\state{0cm}{-13.5cm}{1r, 1b, 6r, 6b, 3r, 3b}{1r/6r, 6r/3r, 1r/3r}{}{}{1b/6b, 1b/3b}{1}
		\draw[arrow] (bottom1)--(top);
		\node[vertexw] (bottom1a) at (bottom) {};
		\state{5cm}{-13.5cm}{1r, 1b, 6r, 6b, 3r, 3b}{1r/3r}{1r/6r, 6r/3r}{1b/6b}{6b/3b, 1b/3b}{1}
		\draw[arrow] (bottom1)--(top);
		\node[vertexw] (bottom2a) at (bottom) {};
		\state{10cm}{-13.5cm}{1r, 1b, 6r, 6b, 3r, 3b}{1r/6r}{6r/3r, 1r/3r}{1b/3b}{1b/6b, 6b/3b}{1}
		\draw[arrow] (bottom2)--(top);
		\node[vertexw] (bottom3a) at (bottom) {};
		\state{15cm}{-13.5cm}{1r, 1b, 6r, 6b, 3r, 3b}{}{1r/6r, 1r/3r}{1b/6b, 6b/3b, 1b/3b}{}{1}
		\draw[arrow] (bottom2)--(top);
		\node[vertexw] (bottom4a) at (bottom) {};

		\node (forg1) at (24, -17) {\ForgetVertexNode{1}};
		\draw[diedge] (intr6)--(forg1);
		\state{5cm}{-18cm}{6r, 6b, 3r, 3b}{6r/3r}{}{}{}{1}
		\draw[arrow] (bottom1a)--(top);
		\node[vertexw] (bottom1) at (bottom) {};
		\state{10cm}{-18cm}{6r, 6b, 3r, 3b}{}{6r/3r}{}{6b/3b}{2}
		\draw[arrow] (bottom2a)--(top);
		\draw[arrow] (bottom3a)--(top);
		\node[vertexw] (bottom2) at (bottom) {};
		\state{15cm}{-18cm}{6r, 6b, 3r, 3b}{}{}{6b/3b}{}{1}
		\draw[arrow] (bottom4a)--(top);
		\node[vertexw] (bottom3) at (bottom) {};

		\node (join) at (33, -22) {\JoinNode};
		\draw[diedge] (forg1)--(join);
		\state{0cm}{-23cm}{3r, 3b, 6r, 6b}{3r/6r}{}{}{}{1\cdot 1}
		\draw[arrow] (bottom1)--(top);
		\node[vertexw] (bottom1a) at (bottom) {};
		\state{5cm}{-23cm}{3r, 3b, 6r, 6b}{3r/6r}{}{}{3b/6b}{1\cdot3}
		\draw[arrow] (bottom1)--(top);
		\node[vertexw] (bottom2a) at (bottom) {};
		\state{10cm}{-23cm}{3r, 3b, 6r, 6b}{3r/6r}{}{3b/6b}{}{1\cdot 1+1\cdot 1}
		\draw[arrow] (bottom1)--(top);
		\draw[arrow] (bottom3)--(top);
		\node[vertexw] (bottom3a) at (bottom) {};
		\state{15cm}{-23cm}{3r, 3b, 6r, 6b}{}{3r/6r}{}{3b/6b}{2\cdot 3 + 2\cdot 3}
		\draw[arrow] (bottom2)--(top);
		\node[vertexw] (bottom4a) at (bottom) {};
		\state{20cm}{-23cm}{3r, 3b, 6r, 6b}{}{3r/6r}{3b/6b}{}{1\cdot3}
		\draw[arrow] (bottom3)--(top);
		\node[vertexw] (bottom5a) at (bottom) {};
		\state{25cm}{-23cm}{3r, 3b, 6r, 6b}{}{}{3b/6b}{}{1\cdot 1}
		\draw[arrow] (bottom3)--(top);
		\node[vertexw] (bottom6a) at (bottom) {};

		\begin{scope}[yshift=-28cm]
			\node (forg4) at (33, 1) {\ForgetVertexNode{4}};
			\draw[diedge] (forg4)--(join);
			\state{5cm}{0cm}{3r, 3b, 6r, 6b}{3r/6r}{}{}{}{1}
			\draw[arrow] (top)--(bottom1a);
			\draw[arrow] (top)--(bottom3a);
			\node[vertexw] (bottom1) at (bottom) {};
			\state{10cm}{0cm}{3r, 3b, 6r, 6b}{}{3r/6r}{}{}{3}
			\draw[arrow] (top)--(bottom4a);
			\draw[arrow] (top)--(bottom5a);
			\node[vertexw] (bottom2) at (bottom) {};
			\state{15cm}{0cm}{3r, 3b, 6r, 6b}{}{}{}{3b/6b}{3}
			\draw[arrow] (top)--(bottom2a);
			\draw[arrow] (top)--(bottom4a);
			\node[vertexw] (bottom3) at (bottom) {};
			\state{20cm}{0cm}{3r, 3b, 6r, 6b}{}{}{3b/6b}{}{1}
			\draw[arrow] (top)--(bottom3a);
			\draw[arrow] (top)--(bottom6a);
			\node[vertexw] (bottom4) at (bottom) {};

			\node (intr3b) at (33, -3.5) {\IntroduceVertexNode{3}};
			\draw[diedge] (intr3b)--(forg4);
			\state{0cm}{-4.5cm}{3r, 3b, 4r, 4b, 6r, 6b}{4r/6r, 3r/4r, 3r/6r}{}{}{3b/4b}{1}
			\draw[arrow] (top)--(bottom1);
			\node[vertexw] (bottom1a) at (bottom) {};
			\state{5cm}{-4.5cm}{3r, 3b, 4r, 4b, 6r, 6b}{4r/6r}{3r/4r, 3r/6r}{3b/4b}{}{1}
			\draw[arrow] (top)--(bottom2);
			\node[vertexw] (bottom2a) at (bottom) {};
			\state{10cm}{-4.5cm}{3r, 3b, 4r, 4b, 6r, 6b}{3r/4r}{4r/6r, 3r/6r}{}{4b/6b, 3b/4b}{2}
			\draw[arrow] (top)--(bottom2);
			\node[vertexw] (bottom3a) at (bottom) {};
			\state{15cm}{-4.5cm}{3r, 3b, 4r, 4b, 6r, 6b}{}{4r/6r, 3r/4r}{3b/4b}{4b/6b, 3b/6b}{2}
			\draw[arrow] (top)--(bottom3);
			\node[vertexw] (bottom4a) at (bottom) {};
			\state{20cm}{-4.5cm}{3r, 3b, 4r, 4b, 6r, 6b}{3r/4r}{}{4b/6b}{3b/4b, 3b/6b}{1}
			\draw[arrow] (top)--(bottom3);
			\node[vertexw] (bottom5a) at (bottom) {};
			\state{25cm}{-4.5cm}{3r, 3b, 4r, 4b, 6r, 6b}{}{3r/4r}{4b/6b, 3b/4b, 3b/6b}{}{1}
			\draw[arrow] (top)--(bottom4);
			\node[vertexw] (bottom6a) at (bottom) {};

			\node (forg5) at (24, -8) {\ForgetVertexNode{5}};
			\draw[diedge] (forg5)--(intr3b);
			\state{5cm}{-9cm}{4r, 4b, 6r, 6b}{4r/6r}{}{}{}{1}
			\draw[arrow] (top)--(bottom1a);
			\draw[arrow] (top)--(bottom2a);
			\node[vertexw] (bottom1) at (bottom) {};
			\state{10cm}{-9cm}{4r, 4b, 6r, 6b}{}{4r/6r}{}{4b/6b}{2}
			\draw[arrow] (top)--(bottom3a);
			\draw[arrow] (top)--(bottom4a);
			\node[vertexw] (bottom2) at (bottom) {};
			\state{15cm}{-9cm}{4r, 4b, 6r, 6b}{}{}{4b/6b}{}{1}
			\draw[arrow] (top)--(bottom5a);
			\draw[arrow] (top)--(bottom6a);
			\node[vertexw] (bottom3) at (bottom) {};

			\node (intr6b) at (24, -12.5) {\IntroduceVertexNode{6}};
			\draw[diedge] (intr6b)--(forg5);
			\state{0cm}{-13.5cm}{5r, 5b, 4r, 4b, 6r, 6b}{5r/4r, 4r/6r, 5r/6r}{}{}{5b/4b, 5b/6b}{1}
			\draw[arrow] (top)--(bottom1);
			\node[vertexw] (bottom1a) at (bottom) {};
			\state{5cm}{-13.5cm}{5r, 5b, 4r, 4b, 6r, 6b}{5r/4r}{4r/6r, 5r/6r}{5b/6b}{5b/4b, 4b/6b}{1}
			\draw[arrow] (top)--(bottom2);
			\node[vertexw] (bottom2a) at (bottom) {};
			\state{10cm}{-13.5cm}{5r, 5b, 4r, 4b, 6r, 6b}{5r/6r}{5r/4r, 4r/6r}{5b/4b}{4b/6b, 5b/6b}{1}
			\draw[arrow] (top)--(bottom2);
			\node[vertexw] (bottom3a) at (bottom) {};
			\state{15cm}{-13.5cm}{5r, 5b, 4r, 4b, 6r, 6b}{}{5r/4r, 5r/6r}{5b/4b, 4b/6b, 5b/6b}{}{1}
			\draw[arrow] (top)--(bottom3);
			\node[vertexw] (bottom4a) at (bottom) {};

			\node (intr5) at (19, -17) {\IntroduceVertexNode{5}};
			\draw[diedge] (intr5)--(intr6b);
			\state{5cm}{-18cm}{4r, 4b, 5r, 5b}{4r/5r}{}{}{4b/5b}{1}
			\draw[arrow] (top)--(bottom1a);
			\draw[arrow] (top)--(bottom2a);
			\state{10cm}{-18cm}{4r, 4b, 5r, 5b}{}{4r/5r}{4b/5b}{}{1}
			\draw[arrow] (top)--(bottom3a);
			\draw[arrow] (top)--(bottom4a);

			\node (intr4) at (29, -17) {\IntroduceVertexNode{4}};
			\draw[diedge] (intr4)--(intr5);
		\end{scope}

		\begin{scope}[xshift=33cm,yshift=-9cm]
			\node (forg6) at (0,0) {\ForgetVertexNode{6}};
			\draw[diedge] (join)--(forg6);
			\state{-1.5cm}{1.5cm}{3r, 3b}{}{}{}{}{22}
		\end{scope}
	\end{tikzpicture}
	\caption{The computation of the cache function for the nice tree decomposition
		from \Cref{fig:nice_tree_decomposition} using the same convention as in \Cref{fig:states_and_colorings}.
		The value of the only state of the root node is~22,
		hence the graph has 20 NAC-colorings.
		The states consistent with the two monochromatic colorings are always the first one and the last one.
	}
	\label{fig:fpt_example}
\end{figure}
\begin{remark}
	\label{rem:fpt_leaf_node}
	If a node $t\in T$ is such that $G_t$ has no edge,
	then $c[t,s]=0$ for all $s\in S_t$.
	In particular, this is the case for the unique state of each leaf node.
\end{remark}

\begin{lemma}%
	\label[lemma]{lemma:fpt_forget_vertex_node}
	Let $t \in T$ be a \ForgetVertexNode{v} node with the only child $t' \in T$.
	If $s \in S_t$, then
	\[
		c[t, s] =
		\sum_{\substack{(H'_\red, H'_\blue) \in S_{t'} \\ (H'_\red - v, H'_\blue - v) = s}}
		c[t', (H'_\red, H'_\blue)]\,.
	\]
\end{lemma}
\begin{proof}
	We have that $G_t = G_{t'}$, hence $G_t$ and $G_{t'}$ have the same set of colorings without almost cycles.
	From the construction of consistent states in \Cref{rem:consistent_state},
	a coloring without almost cycles $\delta$ is consistent with state $s\in S_t$
	if and only $\delta$ is consistent with state $(H'_\red, H'_\blue)\in S_{t'}$
	such that $(H'_\red - v, H'_\blue - v)=s$, the statement follows.
\end{proof}

\begin{definition}
	\label{def:closure}
	Given a graph $H$ such that each edge is either solid or dashed,
	its \emph{closure} is the multigraph with vertices $V(H)$ such that
	there is solid edge $uv$ if and only if $u$ and $v$ are connected by a solid path in $H$
	and there is dashed edge $uv$ if and only if $u$ and $v$ are connected by a path in $H$
	with exactly one dashed edge.
	The closure is \emph{valid} if it is actually a graph, namely, no two vertices
	are connected by a solid and dashed edge at the same time.
\end{definition}

\begin{lemma}%
	\label[lemma]{lemma:fpt_introduce_vertex_node}
	Let $t \in T$ be an \IntroduceVertexNode{v} node such that $G_t$ has an edge and
	$t' \in T$ be the only child of $t$.
	Let $E_v$ be the edges incident to~$v$ in $G_t$.
	Let $s = (H_\red, H_\blue) \in S_t$
	and $E_a^s$, resp.\ $E_a^d$, be the subset of edges of $E_v$ that are also solid,
	resp.\ dashed, edges in~$H_a$ for $a \in \{\red, \blue\}$.
	Suppose that
	\begin{equation}
		\label{eq:assumption_introduce_vertex}
		E_\red^s = E_\blue^d\,, \quad E_\red^d = E_\blue^s\,, \quad
		E_v \subseteq E(H_\red) \quad \text{and}
		\quad E_v \subseteq E(H_\blue)\,.
	\end{equation}
	If $G_{t'}$ has no edge, then $E_v \neq \emptyset$ and
	\[
		c[t, s] =
		\begin{cases}
			1, & \text{if $H_a$ is the closure of the graph with vertices $X_t$, solid edges $E_a^s$} \\
			   & \qquad \text{and dashed edges $E_a^d$ for both $a \in \{\red, \blue\}$}\,,           \\
			0, & \text{otherwise}.
		\end{cases}
	\]
	If $G_{t'}$ contains an edge, then
	\[
		c[t, s] = \sum c[t', (H'_\red, H'_\blue)]\,,
	\]
	where the sum\footnote{We use the common convention that the sum over the empty set is zero.}
	is over all states $(H'_\red, H'_\blue) \in S_{t'}$ such that
	for both $a \in \{\red, \blue\}$,
	$H_a$~is the valid closure of $H'_a$ with added vertex $v$, solid edges $E_a^s$
	and dashed edges $E_a^d$.
	If~\eqref{eq:assumption_introduce_vertex} does not hold,
	then $c[t, s] = 0$.
\end{lemma}
Notice that $E_v$ can be empty if $G_{t'}$ has an edge.
In this case, $c[t,s] \neq 0$ only if
$v$ is an isolated vertex in both $H_\red$ and $H_\blue$.
Then the sum is over a single state, namely $(H_\red - v, H_\blue -v ) \in S_{t'}$.
\begin{proof}
	Suppose that $s$ is consistent with a coloring without almost cycles of $G_t$.
	The fact we use throughout the proof is that
	$e \in E_v$ is red if and only if
	$e$ is a solid edge in $H_\red$ and dashed edge in $H_\blue$
	(as a red edge is a red path and also an almost blue path).
	Similarly, $e \in E_v$ is blue if and only if
	$e$ is a solid edge in $H_\blue$ and dashed edge in $H_\red$.
	Hence, if $c[t, s] \geq 1$, then~\eqref{eq:assumption_introduce_vertex} holds.

	Suppose that $E(G_{t'}) = \emptyset$.
	Since $G_t$ is assumed to have an edge, $E_v \neq \emptyset$.
	Then every edge coloring of $G_t$ by red and blue is a coloring without almost cycles since $G_t$ has no cycle.
	Suppose that $s$ is consistent with such a coloring without almost cycles.
	The red edges are all solid in $H_\red$ and dashed in $H_\blue$,
	while blue edges are all solid in $H_\blue$ and dashed in $H_\red$.
	Since $G_t$ with isolated vertices removed is a star graph,
	red, resp.\ blue, paths in $G_t$ correspond to solid paths in $H_\red$, resp.\ $H_\blue$,
	while almost red, resp.\ almost blue, paths in $G_t$ correspond to paths with exactly
	one dashed edge in $H_\red$, resp.\ $H_\blue$.
	Thus, $H_\red$ and $H_\blue$ are the stated closures.
	Each almost NAC-coloring is consistent with a different state as it determines
	the type of edges incident to $v$.

	Suppose that $E(G_{t'}) \neq \emptyset$.
	Let $s \in S_t$ be consistent with a coloring without almost cycles $\delta$ of~$G_t$.
	The colors of edges in $E_v$ are determined by $s$ by the fact above.
	The restriction $\delta'$ of $\delta$ to $G_{t'}$ is a coloring without almost cycles.
	Let $(H'_\red, H'_\blue) \in S_{t'}$ be the state consistent with $\delta'$.
	Since the paths avoiding $v$ are the same in $(G_t, \delta)$ and $(G_{t'}, \delta')$,
	$H'_\red$ is a subgraph (taking into account also the type of edges)
	of $H_\red$ and $H'_\blue$ is a subgraph of $H_\blue$.
	Since the solid edges in $E_a^s$ and dashed edges in $E_a^d$ give
	monochromatic and almost monochromatic paths respectively
	passing through $v$ for $a \in \{\red, \blue\}$,
	$H_a$ is the closure described in the statement.
	It is valid since in a~coloring without almost cycles, no two vertices are connected
	by a monochromatic and almost monochromatic path of the same color simultaneously.
	Thus, $c[t,s]$ is at most the stated sum.

	On the other hand, consider a coloring without almost cycles $\delta'$ of $G_{t'}$
	consistent with a state $(H'_\red, H'_\blue) \in S_{t'}$ such that
	$H_a$~is the valid closure of $H'_a$ with added vertex $v$, solid edges $E_a^s$
	and dashed edges $E_a^d$ for both $a \in \{\red, \blue\}$.
	Let $\delta$ be an edge coloring of $G_t$ by red and blue such that its restriction to $G_{t'}$ is $\delta'$
	and $\delta(e)$ is $a\in \{\red, \blue\}$ if $e\in E_a^s$,
	which is well-defined by~\eqref{eq:assumption_introduce_vertex}.
	W.l.o.g., suppose there is an almost red cycle in~$(G_t, \delta)$.
	Since $\delta'$ is a coloring without almost cycles, the cycle has to contain two vertices $w,w' \in X_{t'}$.
	But the corresponding red part of the cycle from $w$ to $w'$
	gives a~solid edge in the closure, while the almost red part gives a dashed edge,
	which is not possible as the closure is assumed to be valid.
	Hence, $\delta$ is a coloring without almost cycles and the opposite inequality for $c[t,s]$ holds as well.
\end{proof}

\begin{lemma}%
	\label[lemma]{lemma:fpt_join_node}
	Let $t \in T$ be a \JoinNode{} node such that $G_t$ has an edge
	and let $t_1, t_2 \in T$ be the two children of $t$.
	Let $s = (H_\red, H_\blue) \in S_t$.
	If $G_{t_1}$, resp.\ $G_{t_2}$, has no edges,
	then $c[t, s] = c[t_2, s]$, resp.\ $c[t, s] = c[t_1, s]$.
	Otherwise,
	\begin{align*}
		c[t, s] & = \sum c[t_1, s_1] \cdot c[t_{2}, s_{2}]\,,
	\end{align*}
	where the sum is over all states $s_1 = (H_\red^1,H_\blue^1)\in S_{t_1}$
	and $s_2 = (H_\red^2,H_\blue^2)\in S_{t_2}$
	such that for both $a\in\{\red, \blue\}$,
	the common edges of $H_a^1$ and $H_a^2$ have the same type and
	$H_a$ is the valid closure of the graph with vertices $X_t$,
	solid edges being the union of the solid edges of $H_a^1$ and $H_a^2$ and
	dashed edges being the union of the dashed edges of $H_a^1$ and $H_a^2$.
\end{lemma}
\begin{proof}
	It holds that $X_t = X_{t_1} = X_{t_2}$,
	and thus $S_t = S_{t_1} = S_{t_2}$.
	Notice that the fact that the common edges of $H_a^1$ and $H_a^2$ have the same type
	guarantees we take the closure of a graph, not a multigraph.
	If $G_{t_1}$ has no edges, then $E(G_t)=E(G_{t_2})$ and the statement follows.

	Suppose that both $G_{t_1}$ and $G_{t_2}$ have an edge.
	Let $\delta$ be a coloring without almost cycles of~$G_t$ consistent with $s$.
	Let $\delta_1$ and $\delta_2$ be the restrictions of $\delta$
	to $G_{t_1}$ and $G_{t_2}$ consistent with
	$s_1 = (H_\red^1,H_\blue^1)\in S_{t_1}$ and
	$s_2 = (H_\red^2,H_\blue^2)\in S_{t_2}$ respectively.
	W.l.o.g., suppose there is a solid edge $uv$ in $H^1_\red$ that is a dashed edge in $H^2_\red$.
	But then there is a~red path from $u$ to $v$ in $(G_{t_1}, \delta_1)$ and
	an almost red path from $u$ to $v$ in $(G_{t_2}, \delta_2)$.
	They yield an~almost red cycle in $(G_t, \delta)$,
	which is not possible.
	Since each path in $(G_{t_1}, \delta_1)$ or $(G_{t_2}, \delta_2)$
	is the same in $(G_t, \delta)$,
	$H^1_a$ and $H^2_a$ are subgraphs (taking into account also the type of edges)
	of $H_a$ for $a \in \{\red, \blue\}$.
	Since concatenation of paths in $(G_{t_1}, \delta_1)$ and $(G_{t_2}, \delta_2)$
	corresponds to taking the closure, $s$ is the stated closure.
	Hence, considering all such colorings $\delta$ with restrictions $\delta_1$ and $\delta_2$,
	we have $c[t, s] \leq \sum c[t_1, s_1] \cdot c[t_2, s_2]$.

	For the opposite inequality, let $s_1 \in S_{t_1}$ and $s_2 \in S_{t_2}$
	satisfy the conditions in the statement.
	We can suppose that $c[t_1, s_1] \cdot c[t_{2}, s_{2}] \geq 1$.
	Hence, let $\delta_1$ and $\delta_2$
	be colorings without almost cycles of $G_{t_1}$ and $G_{t_2}$
	consistent with $s_1$ and $s_2$ respectively.
	Let $uv\in G_t$ for $u,v \in X_t$.
	Since $uv \in E(G_{t_1}) \cap E(G_{t_2})$,
	for $i\in\{1,2\}$, the edge $uv$ is either solid in $H_\red^i$ and dashed in $H_\blue^i$,
	or the other way around.
	Since the common edges of $H_a^1$ and $H_a^2$ have the same type,
	$\delta_1(uv)=\delta_2(uv)$.
	Hence, we can define an edge coloring $\delta$ of $G_t$
	by $\delta(e) = \delta_i(e)$ if $e \in E(G_{t_i})$.
	W.l.o.g., suppose there is an almost red cycle in $(G_t, \delta)$.
	Since $\delta_1$ and $\delta_2$ are colorings without almost cycles,
	the cycle has to contain two vertices $w,w' \in X_t$.
	But the corresponding red part of the cycle from $w$ to $w'$
	gives a solid edge in the closure, while the almost red part gives a dashed edge,
	which is not possible as the closure is assumed to be valid.
	Hence, $\delta$ is a coloring without almost cycles and
	$c[t, s] \geq \sum c[t_1, s_1] \cdot c[t_2, s_2]$.
	This concludes the statement.
\end{proof}

\begin{proof}[Proof of \Cref{thm:treewidth}]
	We can suppose that the graph $G$ has at least one edge, otherwise there is no NAC-coloring.
	We compute values of the cache function $c$ on the tree decomposition $T$ of width $k$
	in a bottom-up manner	by using recursive formulas from \Cref{rem:fpt_leaf_node} and
	\Cref{lemma:fpt_introduce_vertex_node,lemma:fpt_forget_vertex_node,lemma:fpt_join_node}.
	Note that it suffices to store only states that have a non-zero value of the cache function.
	Since $G = G_r$, where $r$ is the root node of $T$,
	and there are two monochromatic colorings as $G$ has an edge,
	the number of NAC-colorings of $G$ is $c[r, ((X_r,\emptyset),(X_r,\emptyset))] - 2$.

	We find an upper bound on the number of states in any node $t$ of $T$.
	Let $s=(H_\red, H_\blue)$ be a state in $S_t$ and $n=|X_t|=|V(H_\red)| = |V(H_\blue)|$.
	Let $\ell$ be the number of possible graphs $H_\red$ that satisfy the conditions of \Cref{def:state}.
	It follows that the connected components of the subgraphs of $H_\red$ obtained
	by removing all dashed edges are complete graphs.
	If there is a dashed edge between two such components with vertex sets $V_1$ and $V_2$,
	then there must be a dashed edge $v_1 v_2$ for all $v_1\in V_1$ and $v_2\in V_2$.
	Namely, $\ell$ is at most the number of graphs whose vertex set forms a partition of~$X_t$.
	The number of partitions of a set with $n$ elements is the Bell number,
	which can be upper bounded by $n^n$.
	Hence, $\ell \leq n^n \cdot 2^{\binom{n}{2}}$.
	Since the state consists of pairs $(H_\red, H_\blue)$ and $n\leq k+1$, we have
	\[
		|S_t|\leq \left(n^n \cdot 2^{\binom{n}{2}}\right)^2 \leq n^{2n} \cdot 2^{n^2}
		\leq (k+1)^{2(k+1)} \cdot 2^{(k+1)^2}\in O(f(k))\,,
	\]
	where $f(k) = k^{2k+2} \cdot 2^{k^2+2k}$.
	The closure according to \Cref{def:closure} can be computed in $O(k^4)$ operations.
	We assume the cache function lookups and stores have constant time complexity.
	We go through all the nodes and state complexity of the operations performed in a node:
	\begin{description}
		\item \IntroduceVertexNode{v}:
		      We use the notation of \Cref{lemma:fpt_introduce_vertex_node}.
		      Let $d$ be the degree of $v$ in~$G_t$.
		      Since $v$ can be adjacent only to vertices in $X_t$
					by the second property in \Cref{def:tree_decomposition}, $d$ is at most $k$.
		      If $G_{t'}$ has no edge, then there are exactly $2^d$ states with non-zero cache function.
		      If there is an edge in~$G_{t'}$, then for each state $s'$ of $t'$ and
		      for every possible combination of dashed/solid edges incident to $v$,
		      we compute the closure $s$ of $s'$ with these edges added,
		      and if it is valid, we increment $c[s,t]$ by $c[t', s']$
		      (or set it to this value if it has not been used yet).
		      So the time complexity is $O(k^4 \cdot 2^k \cdot f(k))$.
		\item \ForgetVertexNode{v}:
		      Using the notation of \Cref{lemma:fpt_forget_vertex_node},
		      it suffices to traverse all states $s'$ of~$t'$ such that $c[t',s']\neq 0$
		      and increment the value of the cache function of the state obtained
		      by removing the vertex $v$.
		      This can be done in~$O(k)$ time
		      per state as there are at most $k$ dashed or solid edges incident to $v$.
		      The total complexity therefore is $O(k\cdot f(k))$.
		\item \JoinNode:
		      We iterate through the product of state spaces of $t_1$ and $t_2$
		      (see \Cref{lemma:fpt_join_node}) and increment the value of the cache function
		      of the corresponding closure whenever it is valid.
		      So the number of closure computations is the square
		      of the state space size.
		      Thus, the complexity is $O(k^{4} \cdot f(k)^2)$.
	\end{description}
	There are $O(k|V(G)|)$ nodes in~$T$.
	The final complexity of the algorithm is therefore
	$O(k^{4k+9} 4^{k^2+2k} \cdot |V(G)|)$, which is upper bounded by the simpler formula in the statement.
\end{proof}

Instead of determining the number of NAC-colorings, one can ask only for the existence,
but require a NAC-coloring, if it exists, as a certificate.
This can be achieved by storing a coloring without almost cycles on $G_t$
for each $s \in S_t, t \in T$, where $c[t, s] > 0$,
while preferring colorings that are not monochromatic.

\section{Algorithms}
\label{sec:alg}

The goal of this section is to propose algorithms to find one or all NAC-colorings of a given graph.
After recalling the solution used in \flexrilog{},
we describe an improvement of the idea of \trcon{} components in \Cref{sec:NACvalid}
and checking whether a coloring is a NAC-coloring in \Cref{sec:small_cycles}.
In \Cref{sec:combining} we sketch the idea of algorithms that
exploit combining NAC-colorings of subgraphs.
We propose heuristics for decomposing into subgraphs and merging strategies
in \Cref{sec:decomposition,sec:merging} respectively.
Although the proposed approaches are written so that all NAC-colorings are generated,
the existence of a NAC-coloring can be checked easily by stopping after finding a first one.

A naive approach to list all NAC-colorings is to consider
all $2^{|E(G)|}$ surjective edge colorings of graph $G$ by red and blue
and check each of them in polynomial time using the following lemma.
We call this check \IsNACColoring{}$(G, E_r, E_b)$.

\begin{lemma}[{\cite[Lemma~2.4]{GLS2019}}]
	\label{lemma:is_NAC_coloring}
	Let $G$ be a graph and $\delta: E(G) \to \{\red, \blue\}$ be a~surjective edge coloring.
	Let $E_r$ and $E_b$ be the red and blue edges of $G$ respectively.
	The coloring~$\delta $ is a NAC-coloring if and only if
	the connected components of $G[E_r]$ and $G[E_b]$\footnote{
		By $G[F]$ for $F\subset E(G)$ we mean the edge-induced subgraph,
		namely, the subgraph of $G$ having the edges $F$ and the vertex set
		being the endvertices of the edges in $F$.}
	are induced subgraphs of $G$.
\end{lemma}

The algorithm implemented in \flexrilog{}
considers only colorings such that \trcon{} components are monochromatic.

\subsection{NAC-valid relations}
\label{sec:NACvalid}

While the idea of \trcon{} components is that 3-cycles are monochromatic,
there are also other cases when two edges have to have the same color, which can further reduce the search space.
Namely, we aim to find a partition of the edge set such that each part is necessarily monochromatic
in every NAC-coloring and the number of parts is as small as possible.

\begin{definition}
	\label{def:NACvalid}
	An equivalence relation $\sim$ on the edge set
	of a graph $G$ is called \emph{NAC-valid}
	if for every NAC-coloring $\delta$ of $G$ it holds that
	\[
		\forall e_1, e_2 \in E(G) :
		e_1 \sim e_2 \Rightarrow \delta (e_{1}) = \delta (e_{2})\,.
	\]
	An equivalence class of a NAC-valid relation is called a \emph{NAC-mono(chromatic) class}.
	The \emph{vertices} of a NAC-mono class $M$ are the vertices of the subgraph $G[M]$.
\end{definition}

Clearly, the relation inducing \trcon{} components is NAC-valid.
The following lemma describes a way how to construct a new NAC-valid relation from another one
with possibly less NAC-mono classes.

\begin{lemma}%
	\label{lemma:two_edges_and_component}
	Let $\sim$ be a NAC-valid relation on $G$ and
	$\sim^\prime$ be a relation on $E(G)$ given by
	$e_{1} \sim^\prime e_{2}$ if and only if
	$e_{1} \sim e_{2}$ or there exists
	a cycle $C$ in $G$ such that $e_{1}, e_{2}$
	are edges in $C$
	and all other edges of $C$ are in the same NAC-mono class of $\sim$.
	Then the reflexive-transitive closure of $\sim^\prime$ is NAC-valid.
\end{lemma}
\begin{proof}
	The condition in \Cref{def:NACvalid} is preserved under taking reflexive-transitive closure,
	so it suffices to check it only for $\sim^\prime$: if $e_1$ and $e_2$ had different colors in some NAC-coloring,
	then $C$ would be an almost cycle since $C - \{e_1,e_2\}$ is monochromatic as the relation $\sim$ is NAC-valid.
\end{proof}

Note that various situations can occur: $e_1$ or $e_2$ is in the same NAC-mono class of~$\sim$ as the rest of $C$,
or none of them is.
If each NAC-mono class of $\sim$ induces a connected subgraph,
then the edges $e_1$ and $e_2$ can be supposed to be incident,
but this is not the case in general.
For instance, one could introduce a NAC-valid relation such that both triangles in a 3-prism subgraph
are in the same NAC-mono class, since they have the same color in all NAC-colorings of the 3-prism.

Notice that the application of \Cref{lemma:two_edges_and_component} corresponds
to merging the NAC-mono classes of $e_1$ and $e_2$.
We propose to use the following NAC-valid relation:
we start with the \trcon{} components and then apply the following two steps
as long as there is some change:
\begin{enumerate}
	\item If there is an edge $uv$ such that $u$ and $v$ are vertices of
	      the same NAC-mono class~$M$, then merge $M$ with the NAC-mono class of $uv$.
	\item If there are edges $uv$ and $vw$ such that $u$ and $w$ are vertices of
	      the same NAC-mono class, then merge the NAC-mono classes of $uv$ and $vw$.
\end{enumerate}
Every \trcon{} component induces a connected subgraph and also the two operations preserve
the fact that each NAC-mono class induces a connected subgraph.
Hence, the resulting partition indeed forms NAC-mono classes
of a NAC-valid relation by \Cref{lemma:two_edges_and_component},
we call them \emph{\trext{} classes}.
The construction can be done in polynomial time
and implemented efficiently using a Union-find data structure,
see also \Cref{alg:create_monochromatic_classes}.

\begin{algorithm}[ht]
	\caption{Create \trext{} classes}%
	\label{alg:create_monochromatic_classes}
	\begin{algorithmic}[1]
		\Require{} $G$

		\Ensure{} $P \gets ()$
		\Comment{} list of \trext{} classes

		\State{} $U \gets \Call{CreateUnionFind}{E(G)}$
		\Comment{} create Union-find data structure

		\ForAll{$e_{1},e_{2},e_{3} \in \Call{FindAllTriangles}{G}$}
		\State{} $U \gets \Call{Join}{U, e_{1}, e_{2}, e_{3}}$
		\EndFor{}

		\State{} $U^\prime \gets \emptyset$
		\While{$U^\prime \not= U$}
		\State{} $U^\prime \gets U$
		\Comment{} as long as changes happen

		\State{} $\lambda : V(G) \to \mathcal{P}(E(G))$
		\Comment{} \trext{} classes a vertex is member of,
		\State{}
		\Comment{} classes are represented by a single edge (from Union-find), defaults to $\emptyset$

		\ForAll{$uv \in E(G)$}
		\State{} $e^\prime \gets \Call{Find}{U, uv}$
		\State{} $\lambda(u) \gets \lambda(u) \cup \{e^\prime\}$
		\State{} $\lambda(v) \gets \lambda(v) \cup \{e^\prime\}$
		\EndFor{}

		\ForAll{$v \in V(G)$}
		\If{$\exists u, w \in N(v) : \lambda(u) \cap \lambda(w) \not= \emptyset$}
		\State{} $U \gets \Call{Join}{U, vu, vw}$
		\Comment{} edges are over a same component
		\EndIf{}
		\EndFor{}
		\EndWhile{}

		\State{} $P \gets \Call{Sets}{U}$

	\end{algorithmic}
\end{algorithm}

When searching for all NAC-colorings naively, we have to test $2^{m-1}$ colorings,
where $m$ is the number of NAC-mono classes (we can fix the color of one class).
Hence, using a good NAC-valid relation can reduce the computation time significantly.

\subsection{Small cycles}
\label{sec:small_cycles}

The check \IsNACColoring{} can be quite computationally expensive
considering it is called for every possible coloring.
An optimization we use is to keep a list of shorter cycles in the graph and
check if they are almost monochromatic cycles before doing the full check.

The check for a single cycle can be done in linear time
in the number of NAC-mono classes using bitwise arithmetic.
One cycle $C$ can reject up to $\frac{1}{|E(C)|}$ colorings,
which can be significant for small cycles.

Given a graph $G$ and $k,\ell\in \mathbb{N}$,
we adapt the idea to NAC-mono classes as follows:
for every NAC-mono class $M$ in $G$,
we consider all edges $uv\in M$ such that there is a~path from $u$ to $v$
in $G - M$ that uses at most $k$ NAC-mono classes.
We pick up to~$l$ cycles constructed from $uv$ with the corresponding
path with the least number of used NAC-mono classes.
In our implementation, we take $k=4$ and $l=2$.

\subsection{Combining NAC-colorings of subgraphs}%
\label{sec:combining}

In some cases, a graph can be decomposed into subgraphs so that
the NAC-colorings of the graph can be obtained directly from the NAC-colorings of the subgraphs.
We introduce notation for this process of combining NAC-colorings of suitable subgraphs.

\begin{definition}
	Let $G$ be a graph with subgraphs $G_1, \dots, G_k$
	s.\ t.\ $\bigcup_{i=1}^k E(G_i) =E(G)$.
	For $1\leq i \leq k$, let $\delta^{i}_{\red}$ and
	$\delta^{i}_{\blue}$ be the two monochromatic colorings of $G_i$.
	The set
	\[ \Bigl\{
		\text{surjective }\delta : E(G) \to \{\blue, \red\}
		\mid\forall i : \delta |_{E(G_i)} \in
		\nac{G_i} \cup \{\delta^{i}_{\blue}, \delta^{i}_{\red}\}
		\Bigr\}
	\]
	is called the \emph{NAC-product} of $G_1, \dots, G_k$ and denoted by $\CP(G_1, \dots, G_k)$.
\end{definition}

Since the restriction of a NAC-coloring to a subgraph is a NAC-coloring or monochromatic,
we have $\nac{G_{1} \cup G_{2}} \subseteq \CP(G_{1}, G_{2})$.
If $G_1, \dots, G_k$ are the connected components of $G$,
then $\nac{G} = \CP(G_1,\ldots, G_k)$.
Since every cycle is contained in a single
block\footnote{A \emph{block} is a bridge or a maximal 2-connected subgraph.}
of a graph, we have that $\nac{G} = \CP(G_1,\ldots, G_k)$
for $G_1,\ldots, G_k$ being the blocks of $G$,
see~\cite[Lemma~5.10]{nac_minimally_rigid} for a formula relating the numbers of NAC-colorings.

In order to design a faster algorithm for searching for all NAC-colorings,
we exploit the fact that for edge-disjoint graphs $G_1$ and $G_2$,
it is straightforward to construct $\CP(G_{1}, G_{2})$ once we know $\nac{G_1}$ and $\nac{G_2}$.
Then, we can get $\nac{G_{1} \cup G_{2}}$ by applying \IsNACColoring{}
to each coloring in $\CP(G_{1}, G_{2})$.
This can significantly reduce the number of \IsNACColoring{} calls
compared to testing all red-blue colorings of $G_{1} \cup G_{2}$.
To apply the idea on a graph $G$, we:
\begin{enumerate}
	\item decompose $G$ into pairwise edge-disjoint subgraphs $G_1, \ldots, G_\ell$, and
	\item compute the NAC-product $\CP(G_1, \dots, G_\ell)$ using $\nac{G_1}, \ldots, \nac{G_\ell}$
	      and filter it to get $\nac{G}$.
\end{enumerate}
In the following two subsections, we discuss possible heuristics
for these two phases.

\subsection{Subgraph decomposition}
\label{sec:decomposition}

As the input, we assume the list of NAC-mono classes
and an integer $k\geq 1$, and we output a list of subgraphs
such that each has at most $k$ NAC-mono classes.
We describe some heuristics for how to split a graph into edge-disjoint subgraphs
that perform well with the idea of the previous subsection.

The strategy with the least overhead costs is to take the chunks of $k$ consecutive
NAC-mono classes in the input list (heuristic \None{}).

Our goal is to maximize the number of suitable cycles in each
subgraph as cycles may form almost cycles and those help us to reduce the search
space.
We call the following heuristic \Neighbors, its pseudocode is in \Cref{alg:neighbors}.
The algorithm aims to mimic breadth-first-search
behavior while respecting NAC-mono classes.
The goal is to find monochromatic classes that are close together.
Let $P$ be the list of NAC-mono classes,
for simplicity we assume that $k$ divides $|P|$.
We will gradually divide the NAC-mono classes into bags,
the output is the subgraphs induced by the edges in each bag.

\begin{algorithm}[ht]
	\caption{Heuristic \Neighbors}%
	\label{alg:neighbors}
	\begin{algorithmic}[1]

		\Require{} $G$
		\Comment{} a graph
		\Require{} $P$
		\Comment{} NAC-mono classes of $G$
		\Require{} $k$, $b \gets |P|/k$
		\Comment{} target size of bags, the number of bags

		\Ensure{} $O \gets (B_1, \ldots, B_b)$
		\Comment{} bags with NAC-mono classes

		\While{$P \ne \emptyset$}
		\State{} $p \gets p \in P$
		\State{} $P \gets P \setminus \{p\}$
		\Comment{} random NAC-mono class

		\State{} $B \gets \underset{B_i \in O}{\arg\min} |B_i|$
		\Comment{} the most empty bag

		\State{} $used \gets \{{v \in M} \mid {M \in B } \}$
		\Comment{} vertices already used by this bag
		\State{} $open \gets  \{w \in N(u) \mid u \in used\} \setminus used $
		\Comment{} neighbors, candidates for addition

		\While{$|open| > 0 \land |B| < k$}

		\State{} $best \gets \underset{u \in open}{\arg\max}|used \cap N(u)|$
		\Comment{} vertices with the largest neighborhood

		\For{$n \in N(best) \land |B| < k$}
		\State{} $p \gets p \in P$ such that $\{best, n\} \in p$
		\State{} $P \gets P \setminus \{p\}$

		\State{} $B \gets B \cup \{p\}$

		\State{} $used \gets used \cup \{v \mid {v \in p} \}$
		\State{} $open \gets  \{w \in N(u) \mid u \in used\} \setminus used $
		\EndFor{}
		\EndWhile{}
		\EndWhile{}
	\end{algorithmic}
\end{algorithm}

We start with $\frac{|P|}{k}$ empty bags
and all vertices of the graph labeled $open$.
First we add a random NAC-mono class from the remaining ones to a bag.
All the vertices of the NAC-mono classes in the bag are denoted by $used$.
We take all $open$ vertices that are neighbors of the vertices in $used$ and assign
them a score. A vertex $best$ with the highest score is chosen.
The NAC-mono classes corresponding to edges connecting the vertex $best$
with vertices in $used$ are then added to the bag while not exceeding its size.
The $used$ set is updated, and the algorithm continues until the bag is full.
If $open$ is empty, this iteration of the search also ends.
A new bag is chosen and we repeat the process.

The first strategy \Neighbors\ takes as the score the number of
neighboring vertices of a vertex in the $used$ set.
The other strategy \NeighborsDegree\ is based
on the first and adds another rule ---
if the numbers of neighbors match, the vertex with lower degree is chosen.

\subsection{Subgraph merging}
\label{sec:merging}

After constructing the list of edge-disjoint subgraphs,
we compute all the NAC-colorings for each of them using
the naive algorithm with NAC-mono classes
and improved check on cycles.
Then the results need to be merged to obtain the NAC-colorings of the original graph
using the NAC-product.
Since checking every coloring in the NAC-product is a~costly operation,
we try to minimize the work that has to be done.
The complexity of the task grows with the size of
the merged subgraphs (as \IsNACColoring{} depends on the size of the graph checked)
and also with the number of NAC-colorings found in~each subgraph.
It is also important to note that if the merged subgraphs
have at most one vertex in common, we get the whole NAC-product.

We describe two strategies.
The first approach \MergeLinear{} is
to take the sequence of subgraphs and merge them one by one.
We merge the first and the second subgraph, then we merge this result with the
third one and so on.
We have tried also a tree-like approach, i.e., to merge consecutive pairs
and then recursively again, but the performance was worse.
Another approach, called \SharedVertices{}, always merges two subgraphs that
have the most vertices in common with the goal of
creating as many new cycles as possible.

\section{Benchmarks}%
\label{sec:benchmarks}

In this section, we analyze the performance of
our implementation, which is written in Python and uses the NetworkX library~\cite{networkx}.
We provide the code in~\cite{nac_code},
which also contains tools required to run and analyze benchmarks,
along with precomputed results and generated graph classes.
An additional notebook is provided for experiments with the algorithm itself.
The NAC-coloring search algorithms are implemented as iterators.
Hence, if one requires only one or a few NAC-colorings,
they can be obtained without generating all of them.

The benchmarks were run on Linux on a laptop with an Intel i7 processor of the 11th generation
with CPython 3.12 and SageMath 10.4.
Throughout this section, we write just rigid instead of $2$-rigid, since no other dimension is considered.

\Cref{tab:allMinRigid} shows the time required for finding all the NAC-colorings
of all minimally rigid graphs with a given number of vertices (generated using Nauty~\cite{nauty}
with Nauty Laman plugin~\cite{nauty_plugin}).
We show the results of the implementation in \flexrilog{}~\cite{flexrilog} (SageMath)
and compare it to our implementation of the same naive algorithm using $\triangle$-connected components.
Next, we display the timings for the naive algorithm with \trext{} classes
described in \Cref{sec:NACvalid} using only \IsNACColoring{}
and then with small cycles optimization from \Cref{sec:small_cycles} (\NaiveCycles).
The last column is for the \NeighborsDegree{} (each subgraph has $k=4$ \trext{} classes)
with \MergeLinear{} merging.

\begin{table}[ht]
	\caption{The time (in seconds) needed to find all NAC-colorings for all minimally rigid graphs with a given number of vertices.}
	\label{tab:allMinRigid}
	\vspace{0.3cm}
	\centering
	\begin{tabular}{ccccccc}
		\hline
		\,$|V(G)|$\, & \,\#graphs\, & \,\textsc{FlexRiL.} & $\triangle$-comps.\, & \trextshort\, & \textsc{NaiveCyc.} & \textsc{NeighDeg.}\, \\
		\hline
		8            & 608          & 14               & 1.09                   & 0.97         & 0.36       & 0.49                   \\
		9            & 7\,222       & 509              & 34                     & 29           & 5.8        & 8.6                    \\
		10           & 110\,132     & 27k              & 1\,725                 & 1\,446       & 151        & 213                    \\
		11           & 2\,039\,273  & -                & -                      & -            & 5\,440     & 6\,650                 \\
		\hline
	\end{tabular}
\end{table}

In \Cref{fig:graph_summary},
we show the relation between the number of \IsNACColoring{} checks that
the naive algorithm would perform compared to our solution.
The values are similar for graphs with few \trext{} classes,
which alongside the additional overhead
explains why the cycles-improved naive search outperformed
the \NeighborsDegree{}\&\MergeLinear{} algorithm in \Cref{tab:allMinRigid}.
This should change quickly for larger graphs.
We can also see how the cycles optimization
reduces the number of more expensive \IsNACColoring{} calls,
since these are called only when the small cycles check described
in \Cref{sec:small_cycles} does not reject the considered edge coloring.

\begin{figure}[ht]
	\centering
	\scalebox{0.72}{\input{./figures/graph_export_check-comparision_exp_edge_num_monochromatic_checks_checks_mean.pgf}}
	\caption{The number of \IsNACColoring{} check calls with respect to the number of \trext{} classes
		on some minimally rigid graphs.}%
	\label{fig:graph_summary}
\end{figure}

In the following benchmarks,
we focus on two main tasks: listing all
NAC-colorings of a~graph and finding any NAC-coloring of a graph.
Each benchmark was run twice and the mean was taken.
The graphs are grouped by the number of \trext{} classes.

We present the performance of listing all NAC-colorings
on larger randomly generated minimally rigid graphs depending on the strategy used.
The graphs were randomly generated using NetworkX and PyRigi~\cite{pyrigi}.
The runtimes are in \Cref{fig:graph_time_minimally_rigid}
and the number of \IsNACColoring{} calls in \Cref{fig:graph_count_minimally_rigid}.
Here you can see that the naive approach gets significantly worse for sixteen \trext{} classes and more.
Between nine and sixteen vertices,
the algorithms we propose are already better than the naive one considering the number of checks called,
but worse in runtime, which is caused by additional overhead.

\begin{figure}[ht]
	\centering
	\scalebox{0.72}{\input{./figures/graph_export_minimally_rigid_random_all_monochromatic_runtime_split_merging_mean.pgf}}
	\caption{Mean running time to find all NAC-colorings for minimally rigid graphs.}%
	\label{fig:graph_time_minimally_rigid}
\end{figure}

\begin{figure}[ht]
	\centering
	\scalebox{0.72}{\input{./figures/graph_export_minimally_rigid_random_all_monochromatic_checks_split_merging_mean.pgf}}
	\caption{The number of checks performed to list all NAC-colorings for minimally rigid graphs.}%
	\label{fig:graph_count_minimally_rigid}
\end{figure}

\Cref{fig:graph_time_globally_rigid_some} presents running times of
searching for a single NAC-coloring on randomly generated \emph{globally rigid} graphs
(these are graphs that generically have a unique realization with the same edge lengths
up to rotation, translation, and reflection).
The probability threshold function from~\cite{glob_rigid_threshold} was used for generating them.
We observed that the number of \trext{} classes is usually much smaller
than the number of \trcon{} components.
Most of the graphs tested have many NAC-colorings, hence finding some of them
is fast for all the approaches.
When it comes to listing all NAC-colorings on globally rigid graphs,
\Cref{fig:graph_time_globally_rigid_all}
shows that our algorithm starts to outperform the naive approach by a significant margin.

\begin{figure}[ht]
	\centering
	\scalebox{0.72}{\input{./figures/graph_export_globally_rigid_threshold_first_monochromatic_runtime_split_merging_mean.pgf}}
	\caption{Mean running time to find some NAC-coloring for globally rigid graphs.}%
	\label{fig:graph_time_globally_rigid_some}
\end{figure}

\begin{figure}[ht]
	\centering
	\scalebox{0.72}{\input{./figures/graph_export_globally_rigid_threshold_all_monochromatic_runtime_split_merging_mean.pgf}}
	\caption{Mean running time to list all NAC-colorings for globally rigid graphs.}%
	\label{fig:graph_time_globally_rigid_all}
\end{figure}

Graphs with no NAC-coloring were also tested, see~\Cref{fig:graph_time_no_nac_coloring}.
These were generated randomly using the probability threshold function given in~\cite{thresholds}.
The \trext{} classes generating algorithm is so effective,
that we managed to find only a few graphs that have no NAC-coloring
while having more than one \trext{} class.
Therefore, these benchmarks were run using \trcon{} components only
as the main algorithm should behave similarly.
Neighbors-based strategies perform similarly to \None{}.
The naive algorithm failed to finish within thirty seconds
for most of the graphs with thirty or more components.

\begin{figure}[h]
	\centering
	\scalebox{0.72}{\input{./figures/graph_export_nac_critical_none_first_triangle_runtime_split_merging_mean.pgf}}
	\caption{Mean running time on NAC-critical graphs without NAC-colorings.}%
	\label{fig:graph_time_no_nac_coloring}
\end{figure}

For more benchmarks, strategies and discussions on the optimal size of subgraphs,
and on the performance of \trext{} classes compared to \trcon{} components,
see the thesis~\cite{bc_thesis} or the notebook in~\cite{nac_code}.

To conclude, we would like to emphasize that the algorithms we have proposed together with their implementation outperform
the naive one in \flexrilog{} by two orders of magnitude.
The implementation was used when preparing~\cite{nac_minimally_rigid}
and our code has been merged into the package PyRigi~\cite{pyrigi}, see method \texttt{Graph.NAC\_colorings}.
The described approaches are suitable for extending to special kinds of
NAC-colorings like \emph{cartesian}
or \emph{rotationally symmetric} NAC-colorings, see for instance~\cite{GL2024}
for the definitions.

\section*{Acknowledgments}
J.\,L.\ was supported by the Czech Science Foundation (GAČR), project No.\ 22-04381L\@.
This work was supported by the Student Summer Research Program 2024 of FIT CTU in Prague.

We would like to thank John Haslegrave for the discussion about probability thresholds
for generating random graphs without NAC-colorings,
Georg Grasegger for adding NAC-colorings to the math documentation of PyRigi, and
Michal Opler for an introduction to parametrized complexity and helpful suggestions.

\bibliography{refs}

@phdthesis{bc_thesis,
  author = {Petr Laštovička},
  title = {NAC-colorings search: complexity and algorithms},
  year = {2025},
  school = {Faculty of Information Technology, Czech Technical University in
            Prague},
  address = {Czech Republic},
  type = {Bachelor thesis},
  url = {https://hdl.handle.net/10467/123519},
}

@article{GLS2019,
  author = "Grasegger, Georg and Legerský, Jan and Schicho, Josef",
  title = "{Graphs with Flexible Labelings}",
  journal = "Discrete \& Computational Geometry",
  pages = "461-480",
  year = "2019",
  doi = "10.1007/s00454-018-0026-9",
}

@misc{nac_minimally_rigid,
  author = "Clinch, Katharine and Garamvőlgyi, Dániel and Haslegrave, John and
            Huynh, Tony and Legerský, Jan and Nixon, Anthony",
  title = "{Stable cuts, NAC-colourings and flexible realisations of graphs}",
  year = "2024",
  doi = {10.48550/arXiv.2412.16018},
}

@article{np_complete,
  author = {Dániel Garamvölgyi},
  title = {Global rigidity of (quasi-)injective frameworks on the line},
  journal = {Discrete Mathematics},
  volume = {345},
  number = {2},
  pages = {112687},
  year = {2022},
  issn = {0012-365X},
  doi = {10.1016/j.disc.2021.112687},
}

@article{generically_rigid_graphs,
  AUTHOR = {Asimow, Leonard and Roth, Ben},
  TITLE = {The rigidity of graphs},
  JOURNAL = {Transactions of the American Mathematical Society},
  VOLUME = {245},
  YEAR = {1978},
  PAGES = {279--289},
  DOI = {10.2307/1998867},
}

@article{laman_1970,
  title = "On graphs and rigidity of plane skeletal structures",
  author = "Gerard Laman",
  year = "1970",
  doi = "10.1007/BF01534980",
  volume = "4",
  pages = "331--340",
  journal = "Journal of Engineering Mathematics",
  number = "4",
}

@article{laman_original_1927,
  author = {Pollaczek-Geiringer, Hilda},
  title = {{Über die Gliederung ebener Fachwerke}},
  journal = {ZAMM - Journal of Applied Mathematics and Mechanics / Zeitschrift
             für Angewandte Mathematik und Mechanik},
  volume = {7},
  number = {1},
  pages = {58-72},
  doi = {10.1002/zamm.19270070107},
  year = {1927},
}

@article{stable_cuts,
  title = {On stable cutsets in claw-free graphs and planar graphs},
  journal = {Journal of Discrete Algorithms},
  volume = {6},
  number = {2},
  pages = {256-276},
  year = {2008},
  doi = {10.1016/j.jda.2007.04.001},
  author = {Van Bang Le and Raffaele Mosca and Haiko Müller},
}

@article{polynomial-min-rigid,
  title = {Pebble game algorithms and sparse graphs},
  journal = {Discrete Mathematics},
  volume = {308},
  number = {8},
  pages = {1425-1437},
  year = {2008},
  doi = {10.1016/j.disc.2007.07.104},
  author = {Audrey Lee and Ileana Streinu},
}

@article{nauty,
  author = "Brendan D. McKay and Adolfo Piperno",
  title = "Practical graph isomorphism, \{II\} ",
  journal = "Journal of Symbolic Computation ",
  volume = "60",
  number = "0",
  pages = "94 - 112",
  year = "2014",
  doi = "10.1016/j.jsc.2013.09.003",
}

@inproceedings{networkx,
  author = {Aric A. Hagberg and Daniel A. Schult and Pieter J. Swart},
  title = {{Exploring Network Structure, Dynamics, and Function using NetworkX}},
  booktitle = {Proceedings of the 7th Python in Science Conference},
  pages = {11 - 15},
  address = {Pasadena, CA USA},
  year = {2008},
  editor = {G. Varoquaux and T. Vaught and J. Millman},
}

@inproceedings{flexrilog,
  author = "Grasegger, Georg and Legerský, Jan",
  title = "{FlexRiLoG---A SageMath Package for Motions of Graphs}",
  editor = "Bigatti, A.{\,}M. and Carette, J. and Davenport, J.{\,}H. and Joswig
            , M. and de Wolff, T.",
  booktitle = "Mathematical Software -- ICMS 2020",
  year = "2020",
  pages = "442--450",
  series = {Lecture Notes in Computer Science},
  volume = {12097},
  doi = {10.1007/978-3-030-52200-1_44},
}

@misc{nauty_plugin,
  author = { Martin Larsson },
  title = {{Nauty Laman plugin}},
  year = {2020},
  publisher = {GitHub},
  journal = {GitHub repository},
  howpublished = {\url{https://github.com/martinkjlarsson/nauty-laman-plugin}},
}

@misc{pyrigi,
  title = {{PyRigi -- a general-purpose Python package for the rigidity and
           flexibility of bar-and-joint frameworks}},
  author = {Matteo Gallet and Georg Grasegger and Matthias Himmelmann and Jan
            Legerský},
  year = {2025},
  doi = {10.48550/arXiv.2505.22652},
}

@book{paramAlg,
  author = {Marek Cygan and Fedor V. Fomin and \L{}ukasz Kowalik and Daniel
            Lokshtanov and D{\'{a}}niel Marx and Marcin Pilipczuk and Micha\l{}
            Pilipczuk and Saket Saurabh},
  title = "{P}arameterized {A}lgorithms",
  year = "2015",
  publisher = "Springer Cham",
  doi = {10.1007/978-3-319-21275-3},
  edition = "1st",
}

@article{Courcelle,
  author = {Bruno Courcelle},
  title = {{The monadic second-order logic of graphs. I. Recognizable sets of
           finite graphs}},
  journal = {Information and Computation},
  volume = {85},
  number = {1},
  pages = {12-75},
  year = {1990},
  doi = {10.1016/0890-5401(90)90043-H},
}

@article{GL2024,
  author = {Grasegger, Georg and Legerský, Jan},
  title = "Flexibility and rigidity of frameworks consisting of triangles and
           parallelograms",
  journal = {Computational Geometry},
  volume = {120},
  pages = {102055},
  year = {2024},
  doi = {10.1016/j.comgeo.2023.102055},
}

@article{Dixon,
  author = {Dixon, Alfred C.},
  title = {{On certain deformable frameworks}},
  journal = {Messenger},
  number = {2},
  volume = {29},
  pages = {1--21},
  year = {1899},
}

@inproceedings{WalterHusty,
  author = {Dominic Walter and {Manfred L.} Husty},
  title = {On a nine-bar linkage, its possible configurations and conditions for
           paradoxical mobility},
  booktitle = {12th World Congress on Mechanism and Machine Science, IFToMM 2007
               },
  year = {2007},
}

@article{Le2003,
  author = {Le, Van Bang and Bert Randerath},
  doi = {10.1016/S0304-3975(03)00048-3},
  journal = {Theoretical Computer Science},
  number = {1-3},
  pages = {463--475},
  title = {{On stable cutsets in line graphs}},
  volume = {301},
  year = {2003},
}

@article{GLSclassification,
  AUTHOR = {Grasegger, Georg and Legerský, Jan and Schicho, Josef},
  TITLE = {On the classification of motions of paradoxically movable graphs},
  JOURNAL = {Journal of Computational Geometry},
  VOLUME = {11},
  YEAR = {2020},
  NUMBER = {1},
  PAGES = {548--575},
  DOI = {10.20382/jocg.v11i1a22},
}

@article{GLSinjective,
  author = {Grasegger, Georg and Legerský, Jan and Schicho, Josef},
  title = {Graphs with Flexible Labelings allowing Injective Realizations},
  doi = {10.1016/j.disc.2019.111713},
  journal = {Discrete Mathematics},
  volume = {343},
  number = {6},
  pages = {Art. 111713},
  year = {2020},
}

@misc{nac_code,
  author = {Laštovička, Petr and Legerský, Jan},
  title = {{Algorithms for NAC-coloring search -- implementation and benchmarks}},
  year = {2025},
  publisher    = {Zenodo},
  doi = {10.5281/zenodo.17594025},
}

@inproceedings{Saxe1979,
  author = {Saxe, James},
  title = {{Embeddability of weighted graphs in k-space is strongly NP-Hard}},
  booktitle = {Proc. 17th Allerton Conf. in Communications, Control, and
               Computing},
  pages = {480--489},
  year = {1979},
}

@inproceedings{Schaefer2013,
  author = "Schaefer, Marcus",
  title = "Realizability of Graphs and Linkages",
  editor = "Pach, J{\'a}nos",
  bookTitle = "Thirty Essays on Geometric Graph Theory",
  year = "2013",
  publisher = "Springer New York",
  pages = "461--482",
  doi = "10.1007/978-1-4614-0110-0_24",
}

@article{tree_width_np_complete,
  author = {Arnborg, Stefan and Corneil, Derek G. and Proskurowski, Andrzej},
  title = {Complexity of Finding Embeddings in a k-Tree},
  journal = {SIAM Journal on Algebraic Discrete Methods},
  volume = {8},
  number = {2},
  pages = {277-284},
  year = {1987},
  doi = {10.1137/0608024},
}

@article{tree_width_approximation,
  author = {Eyal Amir},
  title = {Approximation Algorithms for Treewidth},
  journal = {Algorithmica},
  year = {2010},
  volume = {56},
  pages = {448-479},
  doi = {10.1007/s00453-008-9180-4},
}

@inproceedings{AbelDemainEtAl,
  author = {Abel, Zachary and Demaine, Erik D. and Demaine, Martin L. and
            Eisenstat, Sarah and Lynch, Jayson and Schardl, Tao B.},
  title = {{Who Needs Crossings? Hardness of Plane Graph Rigidity}},
  booktitle = {32nd International Symposium on Computational Geometry (SoCG
               2016)},
  pages = {3:1--3:15},
  series = {Leibniz International Proceedings in Informatics (LIPIcs)},
  year = {2016},
  volume = {51},
  editor = {Fekete, S\'{a}ndor and Lubiw, Anna},
  publisher = {Schloss Dagstuhl -- Leibniz-Zentrum f{\"u}r Informatik},
  doi = {10.4230/LIPIcs.SoCG.2016.3},
}

@article{DLinfinite,
  author = {Sean Dewar and Jan Legerský},
  title = {{Flexing infinite frameworks with applications to braced Penrose
           tilings}},
  journal = {Discrete Applied Mathematics},
  volume = {324},
  pages = {1--17},
  year = {2023},
  issn = {0166-218X},
  doi = {10.1016/j.dam.2022.09.002},
}

@misc{thresholds,
      title={{Sharp thresholds for NAC-colourings and stable cuts in random graphs}},
      author={Katie Clinch and John Haslegrave and Tony Huynh and Anthony Nixon},
      year={2025},
      doi={10.48550/arXiv.2510.05838},
}

@article{glob_rigid_threshold,
author = {Jackson, Bill and Servatius, Brigitte and Servatius, Herman},
title = {The 2-dimensional rigidity of certain families of graphs},
journal = {Journal of Graph Theory},
volume = {54},
number = {2},
pages = {154-166},
doi = {10.1002/jgt.20196},
year = {2007}
}
\end{document}